\documentclass[journal,final]{IEEEtran}
\usepackage{cite}
\usepackage[pdftex]{graphicx}
\usepackage{url,amsmath,amssymb,amsthm,bm,algorithm,algpseudocode,MnSymbol,amsfonts,pifont}

\usepackage{hyperref}
\newtheorem{rem}{Remark}
\newtheorem{prop}{Proposition}

\begin{document}
%
\title{Interaction-Aware Labeled Multi-Bernoulli Filter}
%
%
%

\author{Nida~Ishtiaq,
        Amirali~Khodadadian~Gostar,
        Alireza~Bab-Hadiashar,~\IEEEmembership{Senior~Member,~IEEE} 
        and~Reza~Hoseinnezhad
\thanks{All authors are with School of Engineering, RMIT University, Victoria 3083, Australia, e-mail: \{nida.ishtiaq,amirali.khodadadian,abh,rezah\}\@rmit.edu.au}
\thanks{}
\thanks{Manuscript received mm dd, yyyy; revised mm dd, yyyy.}}

%
%

\markboth{IEEE~Transactions~on~Intelligent~Transportation~Systems,~Vol.~vv, No.~nn, mm~yyyy}%
{Ishtiaq~\MakeLowercase{\textit{et al.}}: Interaction-Aware Labeled Multi-Bernoulli Filter}
%



\maketitle

\begin{abstract}
Tracking multiple objects through time is an important part of an intelligent transportation system. Random finite set (RFS)-based filters are one of the emerging techniques for tracking multiple objects. In multi-object tracking (MOT), a common assumption is that each object is moving independent of its surroundings. But in many real-world applications, target objects interact with one another and the environment. Such interactions, when considered for tracking, are usually modelled by an interactive motion model which is application specific. In this paper, we present a novel approach to incorporate target interactions within the prediction step of a RFS-based multi-target filter, i.e. labelled multi-Bernoulli (LMB) filter. The method has been developed for two practical applications of tracking a coordinated swarm and vehicles. The method has been tested for a complex vehicle tracking dataset and compared with the LMB filter through the OSPA and OSPA$^{(2)}$ metrics. The results demonstrate that the proposed interaction-aware method depicts considerable performance enhancement over the LMB filter in terms of the selected metrics.
\end{abstract}


%
\IEEEpeerreviewmaketitle

\section{Introduction}
%
%
%
%
\IEEEPARstart{O}{ne} of the core requirements of an intelligent transportation system (ITS) is its ability to track surrounding objects. This is done by incorporating multi-object tracking (MOT) within ITS. Although an intelligent vehicle should be able to track the surrounding stationary and non-stationary objects, but the tracking of surrounding vehicles is of primary concern. The number of nearby vehicles, their positions and speeds vary as a vehicle moves along a road. Tracking of a randomly varying number of targets (due to target birth and death or entry and exit) is a challenging problem with its difficulty being compounded if there are possible miss-detections and false alarms (clutter) in sensor measurements~\cite{daronkolaei2008joint}. In early 21$^{\mathrm{st}}$ century, a new class of stochastic multi-object filters was invented by Mahler~\cite{mahler_book,mahler_phd2003,MahlerCPHD}. They were called \textit{random finite set}~(RFS) filters and were designed by treating the multi-object entity as a random set of single-object entities (targets) with random cardinality (number of elements in the set). Various approximations, followed by efficient implementations of RFS filters, were then proposed~\cite{Vo_2005,vo_cphd,Vo_FW_Smoothing_TAES_2012,vo_gmphd, Vo_SMC_PHD, MeMBer_Vo2} and applied in different domains~\cite{Reza_2012, Reza_audio_visual,Reza_visual_tracking}. The latest generation of RFS filters, called \textit{labeled RFS filters} append the label of each target into its single-target state and propagate target labels with their states to directly create \textit{target trajectories}~\cite{Vo2013,LMB_Vo2,reuter2014labeled,8531656,papi2015generalized}. The labeled multi-Bernoulli (LMB) filter~\cite{reuter2014labeled} is of particular interest in this paper and will be explored further in Section~\ref{sec:back}.

A major advantage of using stochastic multi-object filters for multi-target tracking, especially with RFS filters, is that they allow us to directly incorporate the environment-related information about targets (regardless of measurements) such as target birth-~and~death-related information, into the prediction step of the filtering process in a mathematically principled manner. This is more evident in the multi-Bernoulli filter and its labeled version, the LMB filter: The birth process is modeled by a set of possibly existing targets, each distributed around one possible area of target entry, and target death is modeled by a state-dependent probability of survival that is small in possible areas of target exit. Incorporation of extra information, if modeled accurately, is expected to result in improvements in the overall cardinality and state estimation and tracking performance of the filter. Nowadays, machine learning-based methods are becoming increasingly common for many applications, including detection-based tracking. However, they have not been deemed very feasible for model-based measurements as these methods require an immense amount of data for training and require sequential detections, which is not the focus of this paper. A recent work \cite{deepLearningModelBased} highlights this issue and develops a deep learning-based method (which is a class of machine learning) and has comparable performance to the state-of-the-art Bayesian methods like $\delta$-generalized labeled multi-Bernoulli ($\delta$-GLMB) filter \cite{papi2015generalized} for simplistic scenarios. Hence for complex model-based tracking, such as the applications discussed in this paper, Bayesian methods are the preferred choice owing to the many drawbacks of deep learning methods highlighted in \cite{deepLearningSurvey2020}.

\textit{Interactions between targets} and how they may affect their movements are something that have been rarely taken into account in formulating stochastic multi-object filters as solutions for multi-target tracking problems. To facilitate the ease of implementation and simplicity of the developed method, most existing target tracking solutions assume that each object moves independently, with no limitations due to the surrounding objects and the environment. This leads to solutions which may perform poorly in a realistic scenario, specifically where the motion often an object is immensely influenced by its surroundings. Specifically in transportation systems, the motion of each object is highly dependent on the surroundings such as motion of other road users and road constraints. In such applications, slight errors in location estimation can lead to highly hazardous consequences\cite{conti2019soft,conti2021location,saucan2020information,win2018theoretical,bartoletti2014mathematical,meyer2018message}. Therefore, there is a need for fundamental tracking solutions which are capable of identifying such interactions and utilizing them for accurate tracking. 

One existing methodology for catering target interactions is forming target groups according to their locations, motion parameters or other application specific parameters\cite{group_tracking1,group_tracking2}. The targets in a group are considered to be behaving in a similar fashion. However, most methods in this category face implementation issues due to splitting and merging of groups. The formation of target groups also introduces direct dependence between targets, often referred as data association. The identification and formulation of target groups is a tedious task, making these methods both mathematically and computationally expensive. In the target tracking literature, particularly in visual tracking and vehicle tracking fields, there have been several attempts to model the interactions between targets. Those models are dominantly \textit{deterministic}. Examples include the interactive motion-based vehicle tracking~\cite{gardner1996interactive,yang2018vehicle}, car-following and lane-changing models~\cite{kesting2007general,rahman2013review,saifuzzaman2014incorporating,song2017multi}, social and group behaviour models for visual tracking~\cite{pellegrini2009you,pellegrini2010improving,yuan2017tracking}. In addition to these methods, as far as we know, except for the generalized labeled multi-Bernoulli~(GLMB) filter~\cite{papi2015generalized}, in other stochastic multi-object filtering solutions, random variations of target states are assumed to be \textit{independent} from each other, i.e. zero interaction is assumed between targets. Even in GLMB filter, target interactions are not directly modeled and incorporated into the filter. Information about such interactions are of the same type of death and birth information, in the sense that they are not measurement-related and could be incorporated into the prediction step of a Bayesian multi-object filter.

Within stochastic filters, there have been a number of works that their design would consider target interactions and their effect on tracking performance. The most prominent examples are the interactive Kalman filter~\cite{khalkhali2019multi} and the unscented Kalman filter~\cite{leven2009unscented,MOT_app_surv1}, and the multiple model filters~\cite{chen2001tracking,Multiple_Model_GLMB}. Kalman filters are naturally designed as single-object filters and when used for multi-target tracking, they are designed as a stack of filters based on availability of prior knowledge about the number of targets. In the multiple model filters, several possible dynamic models are considered for each target movement, and switching from each model to the next is governed by a state machine and constant or state-dependent transition probabilities.

Another approach that could consider interactions between targets, is \textit{direct tracking from image observations} (also called track-before-detect)~\cite{our_icassp_2010,MOT_application_Vision002}. In such works, the multi-target likelihood is a function of image observation (rather than detection extracted from the image observations) and interactions could be indirectly incorporated into the formulation of the likelihood function. Another approach is to use multiple sensors and fuse the information in such a way that despite treating target movement as independent in the prediction step of the stochastic filter, the comprehensive target-related information used in the update step compensate for lack of information on target interactions and still deliver accurate tracking~\cite{Battistelli2013,GCI-MB,Fantacci-BT,Gostar_CSD_Fusion_SP2019,Li_eta_al_TSP_2018_robust_fusion,li2018partial,li2019cardinality,LI2020107246,Wang2018}. Alternatively, one could use controllable sensors that would be actuated/scheduled/selected to obtain information-rich measurements towards compensation for target interaction-related information~\cite{Gostar_CSD_LMB,Gostar2015_TAES,gostar2016multi,gostar2017a,gostar2017b,wang2018a}.

This paper proposes a novel solution to  \textit{directly} incorporate an accurate model of target interactions into the prediction step of an LMB filter, which is our multi-target filter of choice due to its simplicity, intuitiveness and competitive performance in challenging multi-target tracking applications. The core contribution of this work lies in the general applicability of the proposed method. Using no external information, the filter is capable of identifying certain types of interactions among targets. These identified interactions are used to adjust the estimates obtained by the filter, in turn improving the tracking performance. We have devised a novel target prediction methodology where the interactions affect the predicted target states, rather than external interaction-aware motion modelling. If and when a interaction between two targets is identified, the predicted target state is altered by the filter accordingly. We also show how various intuitive deterministic interaction models can be turned into interaction-aware LMB filters for accurate tracking of numerous targets in very challenging applications. Our experiment involves tracking of a large number of vehicles from aerial photos in a complex multi-road junction area. Significant improvements in tracking performance, in terms of optimal sub-pattern assignment (OSPA) metric~\cite{vo_OSPA_Trans_SP}, are demonstrated.

This paper is a significant extension of our recent work~\cite{gostar2019interactive} where the core idea was presented. Here, we present a mathematical proof for LMB prediction. We also present details of implementation, as well as a number of deterministic interaction models that have been already proposed in the literature for practical tracking applications, along with details of how such models can be properly reformulated into stochastic models to be used within our proposed interaction-aware LMB filtering scheme. The experimental results also present a practical challenging vehicle tracking scenario that unlike the original simulations presented in~\cite{gostar2019interactive}, involve tracking of numerous targets with significant ongoing interactions.

The rest of this paper is organised as follows. Section~\ref{sec:back} introduces a background to labeled RFS filters and the notation used in formulating the relevant techniques. Our proposed solution is presented in Section~\ref{sec:our_method}. One possible approach for Sequential Monte Carlo~(SMC) implementation of the proposed interaction-aware LMB filter is presented in Section~\ref{sec:smc_implementation} which also includes examples of how the solution can be implemented in vehicle tracking applications. In Section~\ref{sec:results}, experimental results are provided and discussed. The paper is concluded in Section~\ref{sec:conc}.

\section{Background and Notation}
\label{sec:back}
The notations and definitions used in this paper are summarised in~Table~\ref{tab:notation_and_definitions}. The rest of this section provides a brief review of the general multi-object filter and its prediction and update steps, and the LMB filter.

\begin{table}[t]
    \caption{Notation and definitions}
    \label{tab:notation_and_definitions}
    \centering
    \begin{tabular}{ll}
        \hline
        \textbf{Symbol} & \textbf{Definition} \\ \hline
         lower-case letters (e.g. $x$,$\bm{x}$) \dotfill & single-object state \\
         upper-case letters (e.g. $X$,$\bm{X}$) \dotfill & multi-object state \\
         blackboard letters (e.g. $\mathbb{X,Z,L}$) \dotfill & spaces for variables or labels \\
         $\pi, \bm{\pi}$ \dotfill & a multi-object density \\
         bold-face letters (e.g. $\bm{x,X,\pi}$) \dotfill & labeled entities \\
         $|X|$\dotfill & The cardinality (number \\
         & of elements) of $X$\\
         $\langle f,g \rangle$ \dotfill & $ \int_{\mathbb{X}} f(x) g(x) dx$ \\
         & inner product of two functions \\
         $[p]^X$ \dotfill & $\prod_{x\in X} p(x)$ \\
         $G_X[h]$ \dotfill & $\int\ [h]^X\,\pi(X)\, \delta X$\\
         & Probability Generating Functional \\
         & (PGFl) of RFS variable $X$\\
         $ f_{k|k-1}(\bm{X}_k|\bm{X})$\dotfill & The multi-object state transition \\ 
         & density from time $k-1$ to time $k$\\
         $ f_{k|k-1}(x|x_{k-1},\ell)$\dotfill & The single-object state transition \\
         & density from time $k-1$ to time $k$\\
         $\mathcal{L}(\bm{x})$\dotfill & The label of $\bm{x}$\\
         $\mathcal{L}(\bm{X})$\dotfill & The set of the labels of all\\
         &  members of $\bm{X}$\\
         $1_{\mathbb{L}}(\ell)$\dotfill & equals 1 if $\ell\in\mathbb{L}$ and zero otherwise\\
         $\delta(x)$\dotfill & The Dirac delta function which is 0\\
         & for $x\neq 0$ and satisfies $\int\!\delta(x)dx=1$\\
         \hline
    \end{tabular}
\end{table}

\subsection{The Bayesian multi-object filter}
Let us denote the labeled multi-object state at time $k$ by $\textbf{X}_k\subset\mathbb{X}$ and the multi-object observation by $Z_k\subset\mathbb{Z}$. Both $\textbf{X}_k$ and $Z_k$ are modeled as random finite sets. The multi-object random set distribution is recursively predicted and updated by the filter. We also denote the labeled multi-object prior density (at time $k-1$) by $\bm{\pi}_{k-1}(\cdot|Z_{1:k-1})$, where $Z_{1:k-1}$ is the collection of finite measurements up to time $k-1$. 

The prediction step of the filter is governed by Chapman-Kolmogorov equation,
\begin{equation}
\label{eq:general_prediction}
\bm{\pi}_{k|k-1}(\bm{X}_k|Z_{1:k-1})=\hspace{-1mm}\int\hspace{-2mm}f_{k|k-1}(\bm{X}_k|\bm{X})\bm{\pi}_{k-1}(\bm{X}|Z_{1:k-1})~\delta\hspace{-0.3mm}\bm{X},
\end{equation}
where the integral is the labeled set integral defined in~\cite{Vo2013}. 

In the update step, Bayes' rule returns the following multi-object posterior:
\begin{equation}
\bm{\pi}_{k}(\bm{X}_k|Z_{1:k})=\frac{g_k(Z_k|\bm{X}_k)\bm{\pi}_{k|k-1}(\bm{X}_k|Z_{1:k-1})}
{\int g_k(Z_k|\bm{X})\bm{\pi}_{k|k-1}(\bm{X}|Z_{1:k-1})\delta \bm{X}},
\end{equation}
where $Z_k$ is generally comprised of some measurements each associated with an object (with some objects possibly missed), and some false alarms or clutter. Both the number of object-related measurements and the number of false alarms randomly vary with time. Hence, $Z_k$ is an RFS with its stochastic variations characterized by a multi-object likelihood function $g_k(Z_k|\bm{X}_k)$.

\subsection{The labeled multi-Bernoulli RFS}
An LMB~RFS is the union of a number of \textit{possibly existing} single-object sets that are assumed \textit{statistically independent}, and is denoted by:
\begin{equation}
    \bm{X} = \bigcup_{\ell\in\mathbb{L}} \bm{X}^{(\ell)}
    \label{eq:LMBRFS}
\end{equation}
where each $\bm{X}^{(\ell)}$ is a labeled RFS representing one possibly existing target. The statistical variations of each $\bm{X}^{(\ell)}$ is characterised by (1) its probability of existence denoted by $r^{(\ell)}$, and (2) its single-object density $p^{(\ell)}(x)$ conditioned on its existence. 

As Reuter et. al~\cite{reuter2014labeled} have shown, the LMB distribution is completely described by its component parameters, i.e. $\bm{\pi} = \{(r^{(\ell)},p^{(\ell)}(\cdot))\}_{\ell\in\mathbb{L}}$. Indeed, given the component parameters, the LMB RFS density is given by
\begin{equation}
\label{eq:LMB_Distribution}
\bm{\pi}(\mathbf{X})=\Delta (\mathbf{X})w (\mathcal{L}(\mathbf{X}))\left[ p\right] ^{\mathbf{X}},  
\end{equation}
where $\Delta(\mathbf{X})$ is included to ensure uniqueness of object labels, and is defined to be one if $|\mathbf{X}| = |\mathcal{L}(\mathbf{X})|$ and zero otherwise. Also we have: 
\begin{eqnarray}
    [p]^{\bm{X}} &=& \prod_{(x,\ell)\in\bm{X}} p^{(\ell)}(x) \\
    w(L) &=& \left[\prod\limits_{i \in (\mathbb{L}-L)}\left( 1-r^{(i)}\right)\right] \left[\prod\limits_{\ell \in L} {1_{\mathbb{L}}(\ell)\ r^{(\ell)}}\right]
\end{eqnarray}
where $w(L)$ is the probability of joint existence of all objects with labels $\ell \in L$ and non-existence of all other labels~\cite{reuter2014labeled}.

An important characteristic of the LMB RFS is its PGFl. Considering that PGFl of the union of independent RFSs equals the product of their individual PGFls~\cite{mahler_book}, for the LMB RFS given in~\eqref{eq:LMBRFS} we have:
\begin{equation}
    G_X[h] = \prod_{\ell\in\mathbb{L}} G_{X^{(\ell)}}[h]
\end{equation}
where the PGFl of each single-Bernoulli component is given by:
\begin{equation}
    \begin{array}{rcl}
         G_{X^{(\ell)}}[h] & = &  \int\  [h]^{\bm{X}}\,\bm{\pi}^{(\ell)}(X)\,\delta\bm{X} \\
         & = & [h]^{\emptyset}\,\bm{\pi}^{(\ell)}(\emptyset) + \\
         && \int\ [h]^{\{x\}}\,\bm{\pi}^{(\ell)}(\{x\})\,d{x}.
    \end{array}
\end{equation}
Noting that 
$$
[h]^{\emptyset}\triangleq 1,\ \bm{\pi}^{(\ell)}(\emptyset) = 1 - r^{(\ell)},\ \bm{\pi}^{(\ell)}(\{x\}) = r^{(\ell)}\, p^{(\ell)}(x),
$$
 we have:
\begin{equation}
    \begin{array}{rcl}
         G_{X^{(\ell)}}[h] & = &  \left(1-r^{(\ell)}\right)+\langle h(\cdot)\, ,\, r^{(\ell)} p^{(\ell)}(\cdot)\rangle.
    \end{array}
\end{equation}
Therefore, the PGFl of an LMB RFS density $\{(r^{(\ell)},p^{(\ell)}(\cdot))\}_{\ell\in\mathbb{L}}$ is given as follows:
\begin{equation}
    G_X[h] = \prod_{\ell\in\mathbb{L}} \left[\left(1-r^{(\ell)}\right)+\langle h(\cdot)\, ,\, r^{(\ell)} p^{(\ell)}(\cdot)\rangle\right].
    \label{eq:LMB-PGFl}
\end{equation}

\subsection{The multi-object system model}
The target birth and death, the application constraints and the measurement information are all encapsulated in the multi-object filter through a \textit{multi-object system model}. This model is comprised of two parts: the \textit{multi-object dynamic} model and \textit{the multi-object measurement} model. The former is employed within the prediction step of the filter, and the latter in the update step. 

The multi-object dynamic model incorporates all the non-measurement type information that exist about the targets, including their possible state transitions (e.g. movements) from each time step to the next, possible regions of their entry to the scene (the birth process) and possible regions of their exit (target death). A model for interactions between targets is best to be incorporated into this part of the multi-object system model. Hence, we provide a brief overview of the model here (which will be expanded further as the proposed solution is presented in the paper).

Target death is usually modeled by a state-dependent (normally a location-dependent) \textit{probability of survival} denoted by ${p_S(\bm{x})}$. The target dynamics itself is usually modeled by a state transition density. Given a target state $\bm{x}_{k-1}$ at time $k-1$, the probability density of the next state of the target (at time $k$) is denoted by $f_{k|k-1}(\bm{x_k}|\bm{x_{k-1}})$. For a survey of most common target dynamic models, readers are referred to~\cite{li2003survey}. 

Given a target state $\bm{x}_{k-1} \in \bm{X_{k-1}}$ at time $k-1$, its behaviour at time $k$ is modelled via a Markov shift, resulting in a Bernoulli RFS $\bm{S}_{k|k-1}(\bm{x}_{k-1})$ with probability of existence of $r = p_S({\bm{x}_{k-1}})$ and density of $p(\cdot) = f_{k|k-1}(\bm{\cdot}|{x}_{k-1},\ell)$. Hence, given a labeled set of targets $\bm{X}_{k-1}$ at time $k-1$, its transition to time $k$ is modeled by the LMB:
\begin{equation}
    \bm{T}_{k|k-1}(\bm{X}_{k-1}) = \bigcup_{\bm{x}\in\bm{X}_{k-1}} \bm{S}_{k|k-1}(\bm{x}).
\end{equation}

The appearance of new objects (the birth process) at time $k$ is modelled by an LMB RFS of spontaneous births, denoted by 
$$\bm{\Gamma}_k = \left\{\left(r^{(\ell_B)},p^{(\ell_B)}(\cdot)\right)\right\}_{\ell_B\in\mathbb{B}_k}$$
where $\mathbb{B}_k$ denotes the space of labels of targets that may be born at time $k$. Consequently, the labeled RFS of multi-object state $\bm{x}_k$ at time $k$ is itself an LMB given by the union
\begin{equation}
    \bm{X}_k = \bm{T}_{k|k-1} \cup \bm{\Gamma}_k.
\end{equation}

The multi-object state transition density $f_{k|k-1}(\bm{X}_{k}|\bm{X}_{k-1})$ is the density of an LMB with the parameters:
\begin{equation}
    \label{eq:general_f_k|k-1}
    \begin{array}{c}
    \left\{
        \left(p_S(\bm{x}),f_{k|k-1}(\cdot|\bm{x})
        \right)
        \right\}_{\bm{x}\in\bm{X}_{k-1}}
        \bigcup \ \left\{\left(r^{(\ell_B)},p^{(\ell_B)}(\cdot)\right)\right\}_{\ell_B\in\mathbb{B}_k}.
    \end{array}
\end{equation}

In an LMB filter, assume that the multi-object prior is an LMB density with parameters 
\begin{equation}
    \label{eq:prior_LMB_parameters}
    \bm{\pi}_{k-1} = \left\{\left(r_{k-1}^{(\ell)}, p_{k-1}^{(\ell)}(\cdot)\right)\right\}_{\ell\in\mathbb{L}_{k-1}}
\end{equation}
where $\mathbb{L}_{k-1}$ denotes the space of target labels at time $k-1$. Note that this space is gradually extended with time as new targets are born, according to:
\begin{equation}
    \mathbb{L}_{k} = \mathbb{L}_{k-1} \cup \mathbb{B}_{k}.
\end{equation}
Substituting from equations~\eqref{eq:general_f_k|k-1} and~\eqref{eq:prior_LMB_parameters} in the Chapman-Kolmogorov equation~\eqref{eq:general_prediction}, leads to the following parametrisation of the predicted multi-object density~\cite{reuter2014labeled}:
\begin{equation}
    \bm{\pi}_{k|k-1} = \left\{\left(r_{k|k-1}^{(\ell)}, p_{k|k-1}^{(\ell)}(\cdot)\right)\right\}_{\ell\in\mathbb{L}_{k-1}}\hspace{-3mm} \bigcup \ \left\{\left(r^{(\ell_B)},p^{(\ell_B)}(\cdot)\right)\right\}_{\ell_B\in\mathbb{B}_k} 
\end{equation}
where 
\begin{eqnarray}
    r_{k|k-1}^{(\ell)} & = & \left\langle p_S^{(\ell)}(\cdot),p^{(\ell)}_{k-1}(\cdot) \right\rangle \ r^{(\ell)}_{k-1} \label{eq:predicted_r}\\
    p_{k|k-1}^{(\ell)}(x) & = & \dfrac{\left\langle p_S^{(\ell)}(\cdot) f_{k|k-1}(x|\cdot,\ell)\, ,\, p^{(\ell)}_{k-1}(\cdot) \right\rangle}{\left\langle  p_S^{(\ell)}(\cdot),p^{(\ell)}_{k-1}(\cdot) \right\rangle}. \label{eq:predicted_p}
\end{eqnarray}
\section{The Interaction-Aware LMB Filter}
\label{sec:our_method}
A major observation from the equations reviewed in Section~\ref{sec:back} is that at any time $k$, target birth is modeled by the birth LMB $\bm{\Gamma}_k$, target death is modeled by the state-dependent probability of survival, $p_S(\bm{x})$, and transition of each target state is modeled by the conditional density $f_{k|k-1}(x|x_{k-1},\ell)$. The core idea presented in this paper is that we can model any prior information into the state transition density, with including the \textit{parameters} of the overall multi-object LMB prior as \textit{additional parameters}. Implementing such a model with LMB densities is straightforward because at time $k$, each possibly existing object with label $\ell$ is separately associated with parameters
$r_{k-1}^{(\ell)}$ and $p_{k-1}^{(\ell)}(\cdot)$. Note that the density parameter is usually replaced with weights, means and covariance matrices (if a Gaussian-mixture implementation is applied) or with particles and their weights (if a sequential Monte-Carlo implementation is applied).

We envisage that at any time $k$, interactions between each target $\ell$ and other targets can be formulated by making its transition density dependent on distribution parameters of all the \textit{other} targets at time $k-1$. 

Let us denote the Bernoulli parameters of each target label $\ell$ at time $k$ by
$$
\pi_k^{(\ell)} \triangleq 
\left(
r_k^{(\ell)},p_k^{(\ell)}(\cdot)
\right)
$$
and denote the set of all such parameters except for label $\ell$ by 
$$
\psi_{k}^{(\ell)} \triangleq \bigcup_{\ell'\in\,\mathbb{L}_{k-1}\backslash \{\ell\}} \pi_{k-1}^{(\ell')}.
$$
An \textit{interaction-aware} single-object transition density is denoted by $f_{k|k-1}(x|x_{k-1},\ell;\psi_{k}^{(\ell)})$. The most straightforward example of how the parameters $\psi_{k}^{(\ell)}$ can be incorporated into the state transition density is by inferring a multi-object estimate at time $k-1$ and use it as the parameters in the transition density. 
There are two common approaches for this inference. The first is to extract only those Bernoulli components whose probabilities of existence are greater than a given threshold (e.g. 0.50) and the second is to calculate the cardinality estimate, $|\hat{\bm{X}}_{k-1}|$, as the sum of all probabilities of existence, then consider only those Bernoulli components whose probabilities of existence are among the $|\hat{\bm{X}}_{k-1}|$-th largest. In both cases, if a Bernoulli component with label $\ell$ is chosen, its state estimate is given by 
$$
\hat{\bm{x}}_{k-1}^{(\ell)} = \int\!\! x\,p^{(\ell)}_{k-1}(x)\,dx
$$
and the interaction-aware single-object state transition density is denoted by $f_{k|k-1}(x|x_{k-1},\ell;\hat{\bm{X}}_{k-1}\backslash\,\hat{\bm{x}}_{k-1}^{(\ell)})$.
\begin{rem}
Note that if $\ell\notin \mathcal{L}(\hat{\bm{X}}_{k-1})$, then we simply have $\hat{\bm{X}}_{k-1}\backslash\,\hat{\bm{x}}_{k-1}^{(\ell)} = \hat{\bm{X}}_{k-1}$.
\end{rem}
In Section~\ref{sec:smc_implementation}, we will present a number of examples to demonstrate how common and intuitive interactions between vehicle targets can be modelled into a previous estimate-dependent transition density $f_{k|k-1}(x|x_{k-1},\ell;\hat{\bm{X}}_{k-1}\backslash\,\hat{\bm{x}}_{k-1}^{(\ell)})$.
\begin{figure}
    \centering
    \includegraphics[width=0.9\columnwidth]{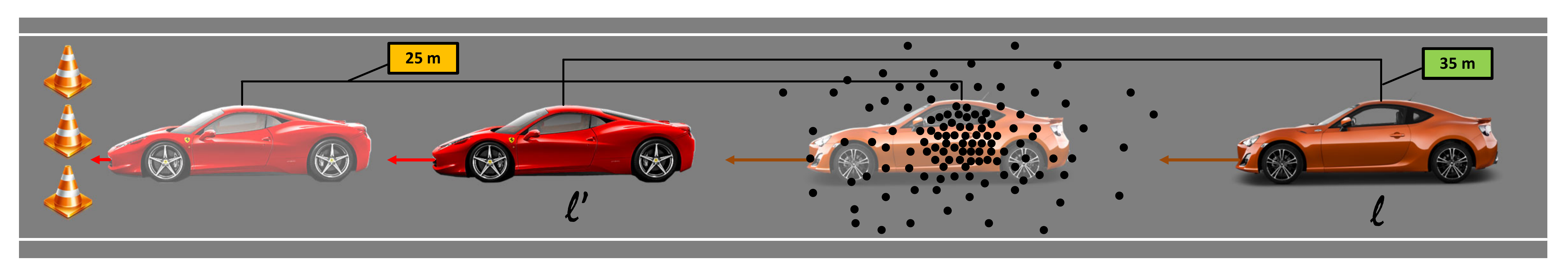}
    \\
    \small (a) \vspace{3mm}
    \\
    \includegraphics[width=0.9\columnwidth]{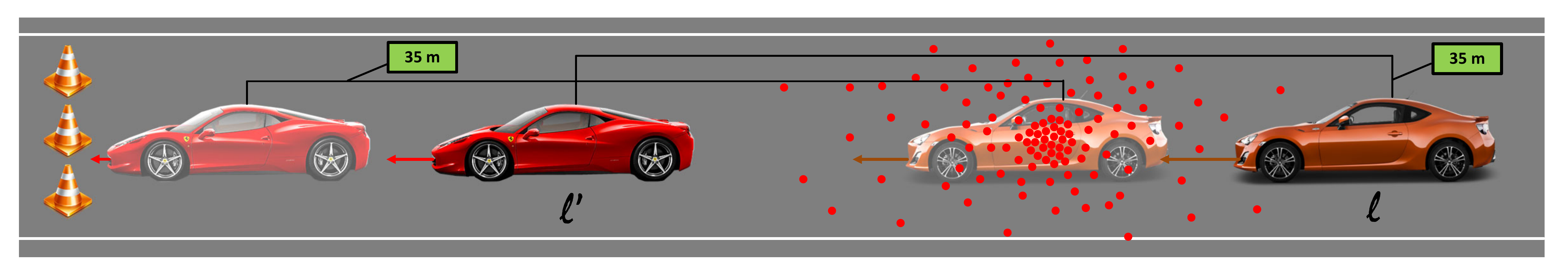}
    \\
    \small (b)
    \caption{An example of predicted particle distributions and expected target locations at next time (faded vehicles) for a vehicle with label $(\ell)$ and its leading vehicle $(\ell')$ when the leading vehicle suddenly reduces speed: (a) without any interactions incorporated, (b) with interactions incorporated.}
    \label{fig:independence_example}
\end{figure}

\begin{rem}
\label{rem:independence}
With incorporating the interaction-aware single-object density, whether in the general form of $f_{k|k-1}(x|x_{k-1},\ell;\psi_{k}^{(\ell)})$ or the more specific form of $f_{k|k-1}(x|x_{k-1},\ell;\hat{\bm{X}}_{k-1}\backslash\,\hat{\bm{x}}_{k-1}^{(\ell)})$, the Bernoulli components of the multi-object LMB at time $k$ can still be assumed \textit{conditionally independent}, given the prior. This is similar to the conditional independence that is commonly assumed for measurements, given the object state.
\end{rem}

The conditional independence of target states discussed in Remark \ref{rem:independence} can be elaborated with the help of a practical example. In Fig.~\ref{fig:independence_example}, two vehicles $(\ell)$(orange car) and ${(\ell)}'$(red car) can be seen moving along a road, maintaining a safe distance between them, which is shown as $35 m$. The arrows ahead of each vehicle represent their velocities, the larger the arrow the greater the velocity. The faded out vehicles show possible next locations of the vehicles. Due to a possible hazard, the leading vehicle ${(\ell)}'$ may need to abruptly stop or slow down significantly by hard braking. This has been depicted by a road closure in Fig.~\ref{fig:independence_example}. It can be seen that the next location of red car has a very small velocity, which should force the vehicle $(\ell)$ to reduce speed as well. In traditional MOT based filtering, the trajectory of the following vehicle $(\ell)$ will not be impacted due to this sudden change in velocity of ${(\ell)}'$, since all objects are assumed to be moving independently, as depicted in Fig.~\ref{fig:independence_example}(a). The predicted particles for the next state of vehicle $(\ell)$ in traditional MOT have been depicted by black circles. It can be seen that they are spread mostly ahead of the vehicle $(\ell)$ according to a motion model, with no affect on the predicted next state due to the sudden change in motion of vehicle ${(\ell)}'$. This would cause the distance between the two vehicles to be significantly reduced, shown as $25 m$ in the figure. In interaction-aware tracking, however, the system makes use of the available information from other possibly interacting targets, which is usually in the form of estimates from the previous time. The ground truth values of target states cannot be directly used by the filter, therefore the target estimates can be used as an approximation of the actual target states (ground truths). By making use of this information, an interaction-aware filter is therefore capable of possibly changing the predicted next state motion information of the following vehicle $(\ell)$ according to any abrupt changes in the motion of ${(\ell)}'$. The possible predicted particles for the next state of vehicle $(\ell)$ have been depicted in Fig.~\ref{fig:independence_example}(b) by red circles. It can be seen that the particles are now more concentrated near a smaller possible velocity for vehicle $(\ell)$, which is the expected actual behaviour of the vehicle. Since the interaction-aware filter uses estimates from the previous time-step to implement such interactive behaviour rather than using the other targets' states from the current time, therefore, the states of targets can still be assumed to be \textit{conditionally independent}. The example shown in Fig.~\ref{fig:independence_example} also depicts the importance of correct estimation of target states. If the position and/or velocity of a vehicle is incorrectly estimated, especially in automated driving systems, it can lead to hazardous situations. Therefore interaction-aware tracking for ITS is of vital importance.

\begin{prop}
\label{prop:int_aware_LMB_prediction}
In an LMB filter, assume that the multi-object prior is an LMB density with parameters given in~\eqref{eq:prior_LMB_parameters} and the single-object state transition density is interaction-aware and parametrised as $f_{k|k-1}(x|x_{k-1},\ell;\psi_{k}^{(\ell)})$. Then, applying the Chapman-Kolmogorov equation~\eqref{eq:general_prediction}, leads to an interaction-aware predicted LMB in which the predicted probabilities of existence are given by~\eqref{eq:predicted_r} and the predicted single-object densities are given by:
\begin{equation}
    p_{k|k-1}^{(\ell)}(x) = \dfrac{\left\langle p_S(\cdot) f_{k|k-1}(x|\cdot,\ell;\psi_{k}^{(\ell)})\, ,\, p^{(\ell)}_{k-1}(\cdot) \right\rangle}{\left\langle  p_S(\cdot),p^{(\ell)}_{k-1}(\cdot) \right\rangle}. \label{eq:interaction_aware_predicted_p}
\end{equation}
\end{prop}

\begin{proof}
Based on the observation made in Remark~\ref{rem:independence} and the PGFl for LMB RFS given in~\eqref{eq:LMB-PGFl}, the PGFl of the possibly surviving targets, conditional on prior multi-object state $\bm{X}_{k-1}$, is given by
\begin{equation}
    \begin{array}{rcl}
        G_{\bm{T}_{k|k-1}}[h|\bm{X}_{k-1}] &=& \prod_{\bm{x}\in \bm{X}_{k-1}} \Big[\big(1-p_S(\bm{x})\big)+\big\langle h(\cdot)\, ,\\
        && \hspace{1.7cm}p_S(\bm{x}) f_{k|k-1}(\cdot| \bm{x};\psi_{k}^{(\ell)})\big\rangle\Big] \\
        & = & \prod_{\bm{x}\in \bm{X}_{k-1}} \Big[\big(1-p_S(\bm{x})\big)+ p_S(\bm{x})\,\big\langle h(\cdot)\, ,\\
        && \hspace{2.6cm} f_{k|k-1}(\cdot| \bm{x};\psi_{k}^{(\ell)})\big\rangle\Big].
    \end{array}
    \label{eq:survivng-PGFl}
\end{equation}
The PGFl of the predicted multi-object RFS is then given by
\begin{equation}
    \begin{array}{rcl}
        G_{\bm{X}_{k|k-1}}[h] &=& \int\ [h]^{\bm{X}} \bm{\pi}_{k|k-1}(\bm{X})\, \delta\bm{X} \\
        &=& \int\ [h]^{\bm{X}} 
        \Big[\int f_{k|k-1}(\bm{X}|\bm{Y};{\Psi}_{k}) \\
        && \hspace{3.4cm}\bm{\pi}_{k-1}(\bm{Y})~\delta\bm{Y}\Big]\,\delta\bm{X} \\
        & = & \int\int\ [h]^{\bm{X}} f_{k|k-1}(\bm{X}|\bm{Y};{\Psi}_{k}) \\
        && \hspace{3.5cm}\bm{\pi}_{k-1}(\bm{Y})~\delta\bm{X}\,\delta\bm{Y} \\
        & = & \int\Big[\int [h]^{\bm{X}} f_{k|k-1}(\bm{X}|\bm{Y};{\Psi}_{k})\,\delta\bm{X} \Big] \\
        &&\hspace{4.1cm}\bm{\pi}_{k-1}(\bm{Y})~\delta\bm{Y}.
    \end{array}
    \label{eq:predicted-PGFl}
\end{equation}
Let us consider to add the birth LMB process at the end of the prediction step. Thus, discarding the birth terms, the term inside the bracket in equation~\eqref{eq:predicted-PGFl} equals $G_{\bm{T}_{k|k-1}}[h|\bm{Y}]$. Thus we have:
\begin{equation}
    \begin{array}{rcl}
        G_{\bm{X}_{k|k-1}}[h] & = & \int G_{\bm{T}_{k|k-1}}[h|\bm{X}] \bm{\pi}_{k-1}(\bm{X})~\delta\bm{X} \\
        & = & \int \prod_{\bm{x}\in \bm{X}} \Big[\big(1-p_S(\bm{x})\big)+ p_S(\bm{x})\,\big\langle h(\cdot)\, ,\\
        && \hspace{1.5cm} f_{k|k-1}(\cdot| \bm{x};\psi_{k}^{(\ell)})\big\rangle\Big]
        \bm{\pi}_{k-1}(\bm{X})~\delta\bm{X}\\
        & = & \int \prod_{\bm{x}\in \bm{X}} h'(\bm{x}) \bm{\pi}_{k-1}(\bm{X})~\delta\bm{X}
    \end{array}
    \label{eq:predicted-PGFl-2}
\end{equation}
where $$
h'(\bm{x}) = \big(1-p_S(\bm{x})\big)+ p_S(\bm{x})\,\big\langle h(\cdot), f_{k|k-1}(\cdot| \bm{x};\psi_{k}^{(\ell)})\big\rangle.
$$
Thus, $G_{\bm{X}_{k|k-1}}[h] = G_{\bm{X}_{k-1}}[h'].$ Noting that $\bm{X}_{k-1}$ is an LMB with parameters $\left\{\left(r_{k-1}^{(\ell)}, p_{k-1}^{(\ell)}(\cdot)\right)\right\}_{\ell\in\mathbb{L}_{k-1}}$, substituting these parameters in~equation~\eqref{eq:LMB-PGFl} results in
\begin{equation}
    \begin{array}{c}
        G_{\bm{X}_{k|k-1}}[h] = \prod_{\ell\in\mathbb{L}_{k-1}} \left[(1-r_{k-1}^{(\ell)})+\big\langle h'(\cdot)\, ,\, r_{k-1}^{(\ell)} p_{k-1}^{(\ell)}(\cdot)\big\rangle\right]
    \end{array}
    \label{eq:predicted-PGFl-3}
\end{equation}
in which the inner product can be further expanded by substituting $h'(\cdot)$ as follows:
\begin{equation}
    \begin{array}{l}
         \big\langle h'(\cdot)\, ,\, r_{k-1}^{(\ell)} p_{k-1}^{(\ell)}(\cdot)\big\rangle = \int \big[\big(1-p_S(\bm{x})\big)+ p_S(\bm{x})\,\big\langle h(\cdot), \\
          \hspace{3.5cm} f_{k|k-1}(\cdot| x,\ell;\psi_{k}^{(\ell)})\rangle\big] r_{k-1}^{(\ell)} p_{k-1}^{(\ell)}(x)\,dx.
    \end{array}
\end{equation}
Substituting the inner product term in~\eqref{eq:predicted-PGFl-3}, followed by some algebraic manipulation (that we omit for the sake of brevity) will lead to
\begin{equation}
    \begin{array}{c}
        G_{\bm{X}_{k|k-1}}[h] = \prod_{\ell\in\mathbb{L}_{k-1}} \left[(1-\rho^{(\ell)}+\big\langle h(\cdot)\, ,\, \rho^{(\ell)} \ q^{(\ell)}(\cdot)\big\rangle\right]
    \end{array}
    \label{eq:predicted-PGFl-4}
\end{equation}
where 
$$\rho^{(\ell)} = \left\langle p_S^{(\ell)}(\cdot),p^{(\ell)}_{k-1}(\cdot) \right\rangle \ r^{(\ell)}_{k-1}$$
and
$$
q^{(\ell)}(x) = \dfrac{\left\langle p_S(\cdot) f_{k|k-1}(x|\cdot,\ell;\psi_{k}^{(\ell)})\, ,\, p^{(\ell)}_{k-1}(\cdot) \right\rangle}{\left\langle  p_S(\cdot),p^{(\ell)}_{k-1}(\cdot) \right\rangle}.
$$
Equation~\eqref{eq:predicted-PGFl-4} matches the mathematical form of the PGFl of an LMB density given by~\eqref{eq:LMB-PGFl}. Therefore, the predicted multi-object RFS, $\bm{X}_{k|k-1}$ turns out to be LMB with parameters 
\begin{eqnarray}
    r_{k|k-1}^{(\ell)} & = & \rho^{(\ell)} = \left\langle p_S^{(\ell)}(\cdot),p^{(\ell)}_{k-1}(\cdot) \right\rangle \ r^{(\ell)}_{k-1} \\
    \nonumber p_{k|k-1}^{(\ell)}(x) & = & q^{(\ell)}(x) 
    = \frac{\left\langle p_S(\cdot) f_{k|k-1}(x|\cdot,\ell;\psi_{k}^{(\ell)})\, ,\, p^{(\ell)}_{k-1}(\cdot) \right\rangle}{\left\langle  p_S(\cdot),p^{(\ell)}_{k-1}(\cdot) \right\rangle}. \\
    && \label{eq:int_aware_LMP_p}
\end{eqnarray}
\end{proof}
\section{Sequential Monte Carlo Implementation}
\label{sec:smc_implementation}
In an SMC implementation (also called particle implementation) of the LMB filter, the density of each Bernoulli component is approximated by particles and their weights as follows~\cite{reuter2014labeled},
\begin{equation}
    p^{(\ell)}_{k-1}(x) = \sum_{j=1}^{J^{(\ell)}_{k-1}} \omega_{k-1,j}^{(\ell)} \ \delta(x-x_{k-1,j}^{(\ell)})
\end{equation}
where 
$
\sum_{j=1}^{J^{(\ell)}_{k-1}} \omega_{k-1,j}^{(\ell)} = 1.
$

To prevent particle death, the particles are commonly \textit{resampled} at the end of previous iteration which leads to a new set of particles with equal weights. In this case we have:
$$
\forall j\in [1:J^{(\ell)}_{k-1}],\ \ \ \omega_{k-1,j}^{(\ell)} = 1\big/{J^{(\ell)}_{k-1}}.
$$
From equation~\eqref{eq:predicted_r} the predicted probability of existence is simply given by:
\begin{equation}
    r_{k|k-1}^{(\ell)} = \eta_S^{(\ell)} r^{(\ell)}_{k-1}
\end{equation}
where
\begin{equation}
    \eta_S^{(\ell)} = \sum_{j=1}^{J^{(\ell)}_{k-1}} \omega_{k-1,j}^{(\ell)} \, p_S^{(\ell)}(x_{k-1,j}^{(\ell)})
\end{equation}
and with equal weights, we have
\begin{equation}
    r_{k|k-1}^{(\ell)} = \left(\frac{\sum_{j=1}^{J^{(\ell)}_{k-1}} p_S^{(\ell)}(x_{k-1,j}^{(\ell)})}{J^{(\ell)}_{k-1}}\right)\,  r^{(\ell)}_{k-1}.
\end{equation}
With the interaction-aware LMB filter proposed in section~\ref{sec:our_method}, the predicted probability of existence is still given by \eqref{eq:predicted_r} -- see Proposition~\ref{prop:int_aware_LMB_prediction}.

To implement equation~\eqref{eq:predicted_p}, the particles are resampled according to an importance density function $q_k^{(\ell)}(\cdot|x_{k-1})$,
$$
\forall j\in [1:J^{(\ell)}_{k-1}],\ \ \ x_{k,j}^{(\ell)} \sim q_k^{(\ell)}\left(\cdot|x_{k-1,j}^{(\ell)}\right)
$$
and then, from~\eqref{eq:predicted_p}, the predicted particle weights are given by:
\begin{equation}
    \label{eq:SMC_weights_original}
    \omega_{k|k-1,j}^{(\ell)} \propto \omega_{k-1,j}^{(\ell)} \ p_S(x_{k-1,j}^{(\ell)})
    \ \frac
    {
    f_{k|k-1}(x_{k,j}^{(\ell)}|x_{k-1,j}^{(\ell)})
    }
    {
    q_k^{(\ell)}(x_{k,j}^{(\ell)}|x_{k-1,j}^{(\ell)})
    }.
\end{equation}
The most common choice for the importance density is the single-object state density itself, in which case we have:
\begin{equation}
    \label{eq:SMC_weights_simpler}
    \omega_{k|k-1,j}^{(\ell)} \propto \omega_{k-1,j}^{(\ell)} \ p_S(x_{k,j}^{(\ell)}).
\end{equation}
In many multi-target tracking applications (e.g. in radar tracking), target death can happen everywhere indiscriminately. In such applications, the probability of survival $p_S$ is constant, and from equation~\eqref{eq:SMC_weights_simpler} the weights of the newly sampled particles remain unchanged.

To implement equation~\eqref{eq:int_aware_LMP_p}, we consider two cases: (i) the interaction-aware single-object density $f_{k|k-1}(x_k|x_{k-1},\ell,\psi_k^{(\ell)})$ can be directly sampled, (ii) the interaction-aware single-object density is modelled as the product of two components, one that is based on a motion model with no interaction included, and one that is focused on modelling target interaction, i.e. 
$$ 
\begin{array}{l}
f_{k|k-1}(x_k|x_{k-1},\ell,\psi_k^{(\ell)}) \propto \bar{f}_{k|k-1}(x_k|x_{k-1},\ell) \times  \\
 \hspace{5.5cm}g_{k|k-1}(x_k|x_{k-1},\psi_k^{(\ell)}).
\end{array}
$$ 
In the former case, similar to the original LMB filter, each particle is propagated to a new one according to 
$$
\forall j\in [1:J^{(\ell)}_{k-1}],\ \ \ x_{k,j}^{(\ell)} \sim f_{k|k-1}\left(\cdot|x_{k-1,j}^{(\ell)},\ell,\psi_k^{(\ell)}\right)
$$
and their weights remain unchanged for a constant $p_S$, otherwise simply scaled by $p_S(\cdot)$ values at each particle -- see~\eqref{eq:SMC_weights_simpler}. In the latter case, not only the particles are propagated according to 
$$
\forall j\in [1:J^{(\ell)}_{k-1}],\ \ \ x_{k,j}^{(\ell)} \sim \bar{f}_{k|k-1}\left(\cdot|x_{k-1,j}^{(\ell)},\ell\right)
$$
but also their weights are scaled (even with constant $p_S$) according to:
\begin{equation}
    \label{eq:SMC_weights_interactive}
    \omega_{k|k-1,j}^{(\ell)} \propto \omega_{k-1,j}^{(\ell)} \ p_S(x_{k-1,j}^{(\ell)})\ g_{k|k-1}(x_{k,j}^{(\ell)}|x_{k-1,j}^{(\ell)},\psi_k^{(\ell)}).
\end{equation}

\subsection{Tracking of coordinated swarm targets}
Consider an application where targets are known to be likely to move in a \textit{coordinated swarm}, i.e. while each target maneuvers, it is likely to maintain its distance from the closest target. Also, targets may die anywhere (hence, the survival probability $p_S$ is constant). Each target state transitions from time $k-1$ to time $k$ according to a linear model
\begin{equation}
\label{eq:linear_model}
    x_k = F x_{k-1} + e_k
\end{equation}
where $e_k$ is a sample of system noise distributed according to a Gaussian with zero mean and covariance matrix $\Sigma$. Thus, the single-target motion is modelled by the state transition density
\begin{equation}
    \bar{f}(x_k|x_{k-1}) = \mathcal{N}(x_k;\, F x_{k-1},\Sigma).
\end{equation}
To implement target's maneuvers based on the above model, we first propagate each particle $j$ according to:
\begin{equation}
\label{eq:linear_model_specific}
    x_{k,j}^{(\ell)} = F x_{k-1,j}^{(\ell)} + e_k.
\end{equation}
Due to the constant probability of survival, the particle weights remain the same.

To model target interactions, we note that the $g_{k|k-1}(x_{k}^{(\ell)}|x_{k-1}^{(\ell)},\psi_{k}^{(\ell)})$ merely affects the particle weights as an importance function -- see~equation~\eqref{eq:SMC_weights_interactive}. Consider $\hat{\bm{X}}_{k-1}^{(\ell)}$ to be the multi-target estimate inferred from LMB parameters at time $k-1$, i.e. from $\psi_{k}^{(\ell)}$. Note that $\hat{\bm{X}}_{k-1}^{(\ell)}$ includes all target state estimates except for the one with label $\ell$. A noise-free next state estimate of each single-target state estimate $\hat{\bm{x}}_{k-1}\in\hat{\bm{X}}_{k-1}$ can be calculated from~\eqref{eq:linear_model} as follows:
\begin{equation}
\label{eq:linear_model_noiseless}
    \bm{\hat{x}}_{k|k-1} = F \bm{\hat{x}}_{k-1}.
\end{equation}
Note that labels stay the same. We denote the ensemble of such predicted estimates by $\hat{\bm{X}}_{k|k-1}$. 

For each Bernoulli component in $\bm{\pi}_{k-1}$ with label $\ell$, we first compute a state estimate 
$$
\hat{x}_{k-1}^{(\ell)} = \sum_{j=1}^{J_{k-1}^{(\ell)}} \omega_{k-1,j}^{(\ell)} x_{k-1,j}^{(\ell)}.
$$ 
Then we find the closest distance from the above estimate to any element of the estimated target states in $\hat{\bm{X}}_{k-1}^{(\ell)}$,
\begin{equation}
\label{eq:dist}
    \hat{d}_{k-1}^{(\ell)} = \min_{\bm{x}\in \hat{\bm{X}}_{k-1}^{(\ell)}} \text{dist}(\hat{x}_{k-1}^{(\ell)}, \bm{x})
\end{equation}
where $\text{dist}(\hat{x}_{k-1}^{(\ell)}, \bm{x})$ is a distance measure. For example, if the Cartesian coordinates of the location of $\hat{x}_{k-1}^{(\ell)}$ and $\bm{x}$ are $[p_{\texttt{x}_{k-1}}^{(\ell)}\ p_{\texttt{y}_{k-1}}^{(\ell)}]^\top$ and $[p_\texttt{x}\ p_\texttt{y}]^\top$, respectively, then the Euclidean distance between the two is given by:
\begin{equation}
\label{eq:dist_estimates}
    \text{dist}(\hat{x}_{k-1}^{(\ell)}, \bm{x}) = \left[\left(p_{\texttt{x}_{k-1}}^{(\ell)} - p_\texttt{x}\right)^2+\left(p_{\texttt{y}_{k-1}}^{(\ell)} - p_\texttt{y}\right)^2\right]^{\frac{1}{2}}.
\end{equation}
We also denote the label of the closest target estimate by $\mathfrak{l}_{k}^{(\ell)}$, and the state estimate itself by $\hat{\bm{x}}^{(\mathfrak{l}_{k}^{(\ell)})}$. The predicted state estimate for this target is a member of $\hat{\bm{X}}_{k|k-1}$ given by
$$
\hat{\bm{x}}_{k|k-1}^{(\mathfrak{l}_{k}^{(\ell)})} = F\, \hat{\bm{x}}^{(\mathfrak{l}_{k}^{(\ell)})}.
$$
Since target interaction is stated as maintaining the distance from closest target, for each particle $x_{k,j}^{(\ell)}$, the closer its distance from $\hat{\bm{x}}_{k|k-1}^{(\mathfrak{l}_{k}^{(\ell)})}$ is to the distance $\hat{d}_{k-1}^{(\ell)}$, the more its weight should be increased. Therefore,
\begin{equation}
    g_{k|k-1}(x_{k,j}^{(\ell)}|x_{k-1}^{(\ell)},\psi_{k}^{(\ell)}) = \mathcal{N}(e_{k,j}^{(\ell)}; 0,\sigma_d^2)
    \label{eq:g-komaki}
\end{equation}
where
$$
e_{k,j}^{(\ell)} = \hat{d}_{k-1}^{(\ell)} - \text{dist}(x_{k,j}^{(\ell)}, \hat{\bm{x}}_{k|k-1}^{(\mathfrak{l}_{k}^{(\ell)})})
$$
and $\sigma_d$ is a user-defined parameter to decide how much change in distance to closest target from time $k-1$ to $k$ is acceptable.

\begin{figure}
    \centering
    \includegraphics[width=0.8\columnwidth]{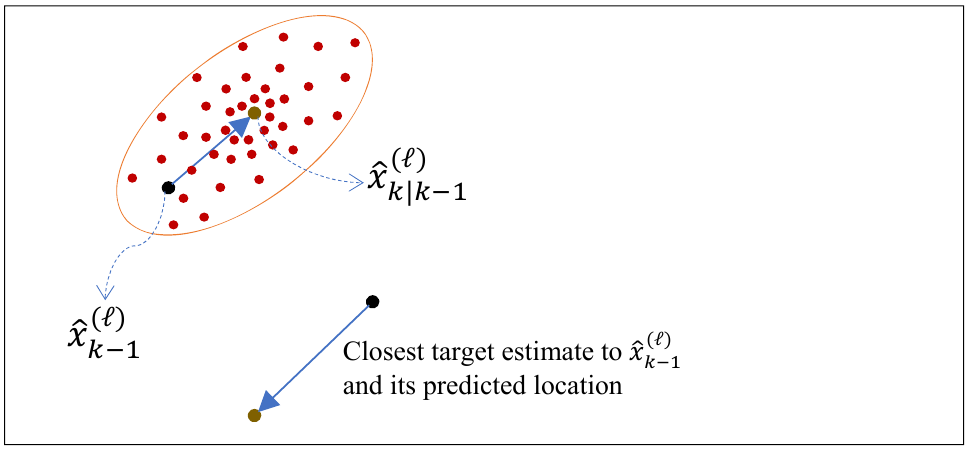}
    \\
    \small (a) \vspace{3mm}
    \\
    \includegraphics[width=0.8\columnwidth]{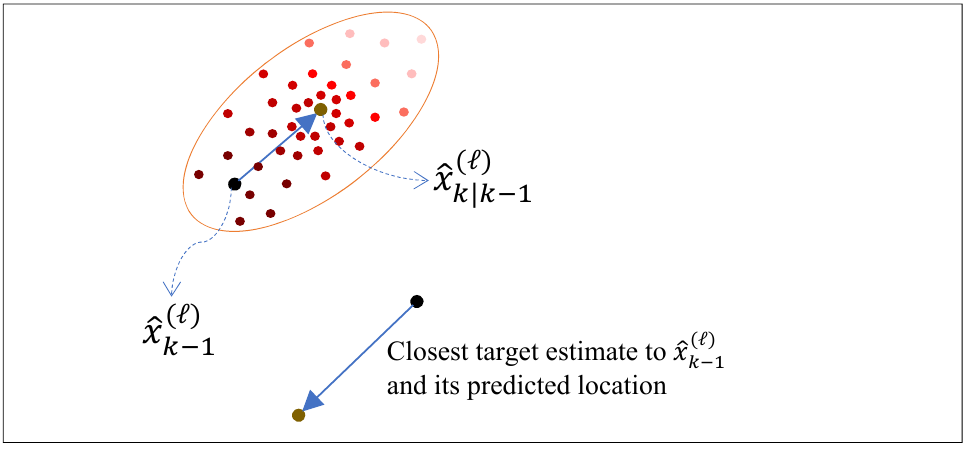}
    \\
    \small (b)
    \caption{An example of particle distributions after prediction for target with label $(\ell)$: (a) without any interactions incorporated, (b) with interactions incorporated.}
    \label{fig:particle_example}
\end{figure}

Figure~\ref{fig:particle_example} shows an example of how particles and their weights would be distributed. With no interactions incorporated into the filter, the particles will be distributed around the next predicted state of the target, all with the same weight, i.e. same shade of red in Fig.~\ref{fig:particle_example}(a). With the interaction-aware filter, however, the weights of the particles change. As the scenario in Fig.~\ref{fig:particle_example} shows, the movements of the closest target to target $(\ell)$ are in such a way that at time $k$, the target $(\ell)$ is more likely to stay at the same location to maintain its distance from the closest target, rather than move according to its motion model. Hence, in Fig.~\ref{fig:particle_example}(b), although the particles are distributed according to the motion model, those closer to the current location of the target $(\ell)$ are assigned larger weights (darker shade of red colour in the figure).

\subsection{Front-vehicle follower model}
Consider another application where numerous vehicles are to be tracked in a busy multi-lane road. The most prominent interaction between vehicles is that every target maintains its distance from the closest target \textit{in its front}. Incorporating this type of interaction is more complicated than interactions between targets in a coordinated swarm. 

Consider the car labeled \ding{172} in the scenario shown in Fig.~\ref{fig:vehicles_example}. Here, incorporating the interaction is not merely maintaining the shortest distance between car \ding{172} and all other cars but only cars that are travelling (i) in the same direction as \ding{172} and (ii) in front of it. This way, assuming that at time $k-1$ all the cars shown in the figure are detected and their labels and state estimates appear within $\hat{\bm{X}}_{k-1}$, care should be taken in the design of the $g_{k|k-1}(\cdot|\bm{x}_{k-1},\psi_k^{(\ell)})$ function to ensure that for car \ding{172}, the closest car (for its distance to be maintained) is not chosen as \ding{176} or \ding{177}, but \ding{173}.

\begin{figure}
    \centering
    \includegraphics[width=0.9\columnwidth]{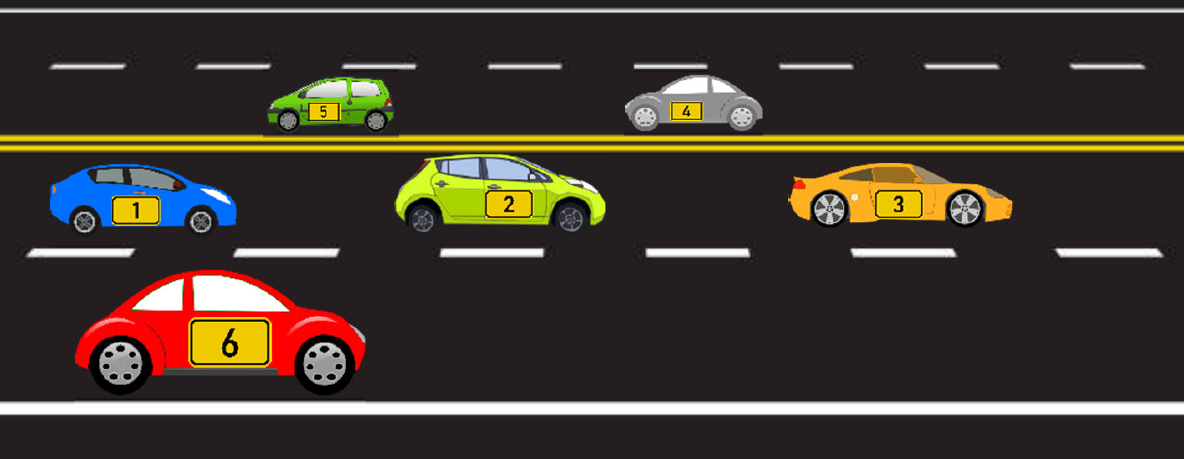}
    \caption{When considering interactions of vehicle 1 with other vehicles, the closest of only vehicles 2 and 3 are to be considered.}
    \label{fig:vehicles_example}
\end{figure}
In the following, we explain how such a $g_{k|k-1}(\cdot|\bm{x}_{k-1},\psi_k^{(\ell)})$ function can be mathematically devised. The main difference here is the consideration of target velocity direction in identifying the interacting targets. As with the coordinated swarm target case, we first propagate particles according to the linear model in~\eqref{eq:linear_model_specific}. For the purpose of this application, let us assume that the linear motion model $F$ is described by the nearly constant velocity (NCV) motion model. After particle propagation, for each unique target label, we identify if there is a possible interaction with a close front vehicle/target. If no such interaction is found, the predicted particle weights are calculated using ~\eqref{eq:SMC_weights_simpler}. However, if a possible interaction is identified, the particle weights are calculated according to  ~\eqref{eq:SMC_weights_interactive}. 

In order to determine if a target/vehicle with label $\ell$ is following another target/vehicle closely, we utilize the estimates from time $k-1$, $\hat{\bm{X}}_{k-1}$. According to the complexity of the road and the general speed of vehicles, we set a threshold on distance between vehicles $d_{\text{th}}$. If the distance is less than threshold for two vehicles, the vehicle in pursuit is expected to be impacted by the motion of its front vehicle. Therefore, we first calculate the distances between the estimated location of a target with label $\ell$, $\hat{x}_{k-1}^{(\ell)}$, and all other target estimates from time $k-1$,  $\hat{\bm{X}}_{k-1}^{(\ell)}$, as described in \eqref{eq:dist_estimates}. The set of targets for which the distance from target $\ell$ falls below $d_{\text{th}}$ are considered to be ``near'' target $\ell$, with their labels collected in the set $L_{\text{near}}^{(\ell)}$:
\begin{equation}
\label{eq:dist_near}    
    L_{\text{near}}^{(\ell)} \triangleq \left\{ \ell\in \mathbb{L}-\{\ell\} ;\ \  \text{dist}(\hat{x}_{k-1}^{(\ell)} , \bm{x}) \leq d_{\text{th}} \right\}.
\end{equation}
Let us denote the state estimates of near targets in $L_{\text{near}}^{(\ell)}$ by $\hat{\bm{X}}_{\text{near}}^{(\ell)}$ as a subset of $\hat{\bm{X}}_{k-1}^{(\ell)}$. Vehicles in the same lane are bound to be moving in the same direction. The use of NCV motion model indicates that each target state is comprised of the target location and velocity in $\texttt{x}$ and $\texttt{y}$ coordinates, 
\begin{equation}
\bm{x} = [p_\texttt{x}\ \ \dot{p}_\texttt{x}\ \ p_\texttt{y}\ \  \dot{p}_\texttt{y}]^\top.
\label{eq:single-target-state}
\end{equation}

The velocity components $\hat{v}_{\bm{x}} =[\dot{p}_\texttt{x}\ \ \dot{p}_\texttt{y}]^\top$ of the estimates can be used to find the angle between velocity vectors of the target $\ell$ and all the near vehicles with labels in $L_{\text{near}}^{(\ell)}$. If this angle is smaller than a threshold $\alpha_{\text{th}}$, the targets can be assumed to have similar direction of motion. This is shown in Fig.~\ref{fig:vehicles_velocity_angle} where the vehicles with labels $\ell$ and $\ell_{1}^{'}$ are moving in same direction hence the angle between their velocity vectors should be very small (even after consideration that noise makes the velocity direction seem slightly varying), as opposed to the angle between velocities of vehicles $\ell$ and $\ell_{2}^{'}$.
Accordingly, we may update the set of near targets to $\ell$, $L_{\text{near}}^{(\ell)}$, as follows:
$$
 L_{\text{near}}^{(\ell)} \leftarrow \left\{ \ell' \in L_{\text{near}}^{(\ell)};\ \measuredangle \left(\hat{v}_{x_{k-1}}^{(\ell)} , \hat{v}_{x_{k-1}}^{(\ell')}\right) \leq \alpha_{\text{th}}\right\}
$$
where
$$
\measuredangle \left(\hat{v}_{x_{k-1}}^{(\ell)} , \hat{v}_{x_{k-1}}^{(\ell')}\right) \triangleq \cos^{-1}
\left(\frac{
\left({\hat{v}_{x_{k-1}}^{(\ell)}}\right)^\top \hat{v}_{x_{k-1}}^{(\ell')}}
{
\left| \hat{v}_{x_{k-1}}^{(\ell)} \right|\ 
\left| \hat{v}_{x_{k-1}}^{(\ell')} \right|
}\right).
$$
\begin{figure}
    \centering
    \includegraphics[width=0.9\columnwidth]{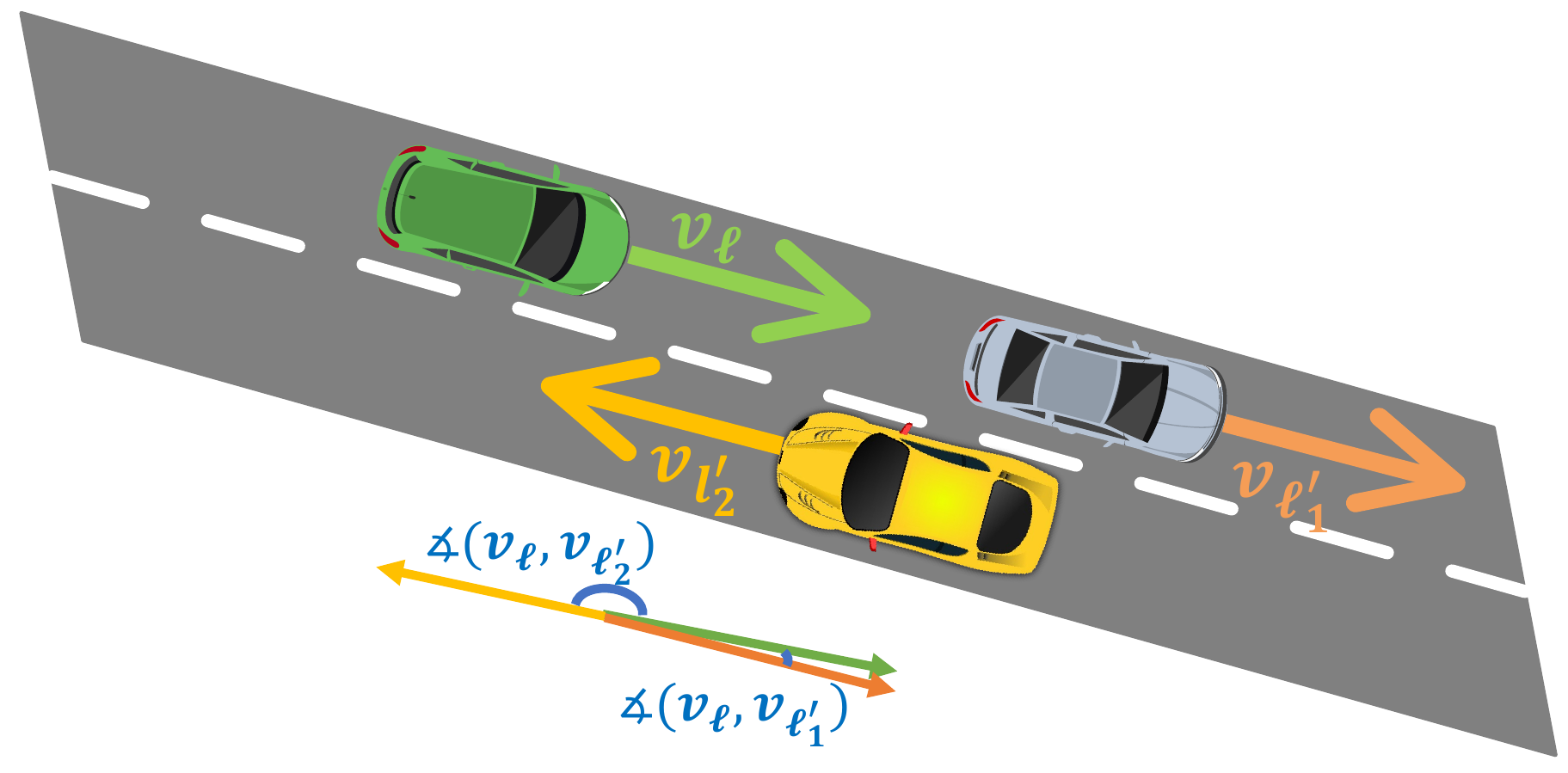}
    \caption{An example demonstrating how we can determine whether ``near'' vehicles are moving in the same or opposite direction by investigating the angle between velocity vectors.}
    \label{fig:vehicles_velocity_angle}
\end{figure}

The above condition excludes all vehicles travelling in a direction different from vehicle $\ell$ from possibly interacting targets. However, it may still include a vehicle travelling closely behind vehicle $\ell$. Therefore, we check the position vector of each vehicle in the updated set of ``near'' labels, relative to vehicle $\ell$: 
\begin{equation}
    \hat{p}^{(\ell\rightarrow \ell')} \triangleq \left[
    \left(\hat{p}_{\texttt{x}_{k-1}^{(\ell')}} - \hat{p}_{\texttt{x}_{k-1}^{(\ell)}}\right)
    \ \ \ \ 
    \left(\hat{p}_{\texttt{y}_{k-1}^{(\ell')}} - \hat{p}_{\texttt{y}_{k-1}^{(\ell)}}\right)
    \right]^\top
\end{equation}
and investigate whether it is in around the same direction as the vehicle velocity vector ${\hat{v}_{x_{k-1}}^{(\ell)}}$ or not. This is simply done by comparing the x and y components of the position of vehicles $\ell$ and $\ell'$, according to the x and y components of ${\hat{v}_{x_{k-1}}^{(\ell)}}$, respectively. Furthermore, we ascertain if a ''near'' vehicle is actually moving in front of the vehicle $\ell$ and not behind it. Figure~\ref{fig:front_or_beihnd} shows an example to demonstrate how we run this check. Consider the two ``near'' vehicles $\ell'_1$ and $\ell'_2$. For $\ell'_2$ which is behind the target vehicle $\ell$, the angle between its relative position vector $\hat{p}^{(\ell\rightarrow \ell'_2)}$ and the target vehicle velocity vector $\hat{v}_{x_{k-1}}^{(\ell)}$ is very large (close to 180$^\circ$), while this angle is small for the vehicle $\ell'_2$ that is in front of $\ell$. Thus, we can shrink the set of ``near'' vehicle labels further as follows:

$$
 L_{\text{near}}^{(\ell)} \leftarrow \left\{ \ell' \in L_{\text{near}}^{(\ell)};\ \measuredangle \left(\hat{v}_{x_{k-1}}^{(\ell)} , \hat{p}^{(\ell\rightarrow \ell')}\right) \leq \beta_{\text{th}}\right\}
$$
where $\beta_{\text{th}}$ is a threshold. 

\begin{figure}
    \centering
    \includegraphics[width=0.99\columnwidth]{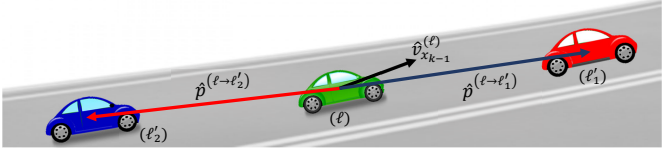}
    \caption{An example demonstrating how we determine whether a ``near'' vehicle is in front or behind by investigating the angle between location and velocity vectors.}
    \label{fig:front_or_beihnd}
\end{figure}

The last step is to find the closest vehicle that is still included as a ``near'' vehicle in $L_{\text{near}}^{(\ell)}$ (i.e. the vehicle that is driven right in front of $\ell$ and is not too far from it, hence interacting with it):
\begin{equation}
\label{eq:min_dist_front_vehicle}
    \mathfrak{l}_{k}^{(\ell)} = \underset{\ell'\in L_{\text{near}}^{(\ell)}}{\arg\min}\ \ \text{dist}(\hat{x}_{k-1}^{(\ell)}, \hat{x}_{k-1}^{(\ell')}).
\end{equation}
Once the label of the interacting target, $\mathfrak{l}_{k}^{(\ell)}$, is found, similar to section~\ref{sec:smc_implementation}.A, first its distance is calculated according to $\hat{d}_{k-1}^{(\ell)} = \text{dist}(\hat{x}_{k-1}^{(\ell)}, \hat{x}_{k-1}^{(\mathfrak{l}_{k}^{(\ell)})})$, then its estimate is propagated to time $k$ via $\hat{\bm{x}}_{k|k-1}^{(\mathfrak{l}_{k}^{(\ell)})} = F\, \hat{\bm{x}}^{(\mathfrak{l}_{k}^{(\ell)})}$ then for each particle $x_{k,j}^{(\ell)}$, the error term 
$e_{k,j}^{(\ell)} = \hat{d}_{k-1}^{(\ell)} - \text{dist}(x_{k,j}^{(\ell)}, \hat{\bm{x}}_{k|k-1}^{(\mathfrak{l}_{k}^{(\ell)})})$ is calculated and used within the density term in equation~\eqref{eq:g-komaki}, which is finally used to increase or decrease the weights of those particles in effect of interaction.
\section{Experimental Results}
\label{sec:results}
The proposed method has been tested on a complex road intersection scenario, where interactions between targets are identified using the steps outlined in Section \ref{sec:smc_implementation}.B.  An aerial video of the Swindon M4 motorway junction (J-16) in the U.K has been selected, which is shot by a high functioning drone camera (DJI inspire 2). A sample image from the dataset is shown in Fig. ~\ref{fig:dataset_scenario}. The dataset contains a large number of vehicles at any time, usually more than 60 vehicles per frame. The intersection has different challenges like heavy, fast traffic as well as traffic lights which control the motion of vehicles on the road. The size of the target vehicles is very small in the images, making it a challenging tracking task. The vehicles undergo various maneuvers like varying speeds, turns and lane changes. The proposed method is capable of identifying all the mentioned maneuvers, rather than ignoring some maneuvers for simplicity. Furthermore, we identify target interactions based on distance between two vehicles, without using any dataset-specific information. A total of 245 consecutive frames have been selected from the video for our implementation. The camera is assumed to be stationary for all images. The ground truth for 200 vehicles has been manually annotated for the dataset as part of this research, using the MATLAB's ground truth labeler application. The measurements have been obtained by applying a noisy detector to the ground truth. The LMB filter has been used for tracking where there is no interaction, and the interactive weight update from ~\eqref{eq:SMC_weights_interactive} has been implemented for vehicles with identified interactions. In addition to the performance comparison discussed in this section, a results video is provided as supplementary material to this paper. The video shows each frame of the dataset at a frame rate of 2 frames per second (fps). For each frame, ground truth has been represented with green dots on the centre point of the target vehicle and estimates are represented with red boxes around the vehicles, with target labels written in red to identify any possible label switching. 

As the estimated target velocities are used for identifying any possible interactions, the interaction model is implemented after $5$ time frames so that the estimated velocities have been well adjusted according to the measurement and motion models. The NCV motion model is used in which the state of each target is described according equation~\eqref{eq:single-target-state}. In this model, the single-target state transition density is given by:
\begin{equation}
		f\left( x_k | x_{k-1} \right) = \mathcal{N} \left( x_k ; F x_{k-1}, Q \right)  
	\end{equation}
	where $F$ is the state transition matrix and $Q$ is the process noise covariance matrix:
	\begin{equation}
		F = \text{diag}\left(A,A\right) \qquad Q = \text{diag}\left(B,B\right) 
	\end{equation}
and $A \text{ and } B$ are matrices given by:
	\begin{equation}
		A = 
		\begin{bmatrix}
			1 & \quad T \\
			0 & \quad 1
		\end{bmatrix} \qquad B = \sigma ^2_{\text{motion}}
		\begin{bmatrix}
			T^3 / 3 & \quad T^2 / 2 \\
			T^2 / 2 & \quad T
		\end{bmatrix}
	\end{equation}
in which $T$ is the time step, which is set equal to $1$\;s. The $\sigma_{motion}^2$ is 7\;pixels/s$^2$. The probability of target detection $P_d$ is $0.995$ and probability of survival $P_s$ is $0.99$. The observation noise is a zero mean Gaussian with variance 3\;pixels/s$^2$. The target birth is generated at all entry points on the road in the scenario, as well as some central points at the intersection. The birth model used is an LMB~RFS according to $\{ (r_{B}^{(\ell)} , p_{B}^{(\ell)}) \}_{\ell \in \mathbb{B}}$ where $r_{B}^{(\ell)} = 0.2$ for all birth components.

\begin{figure}
    \centering
    \includegraphics[width=0.9\columnwidth]{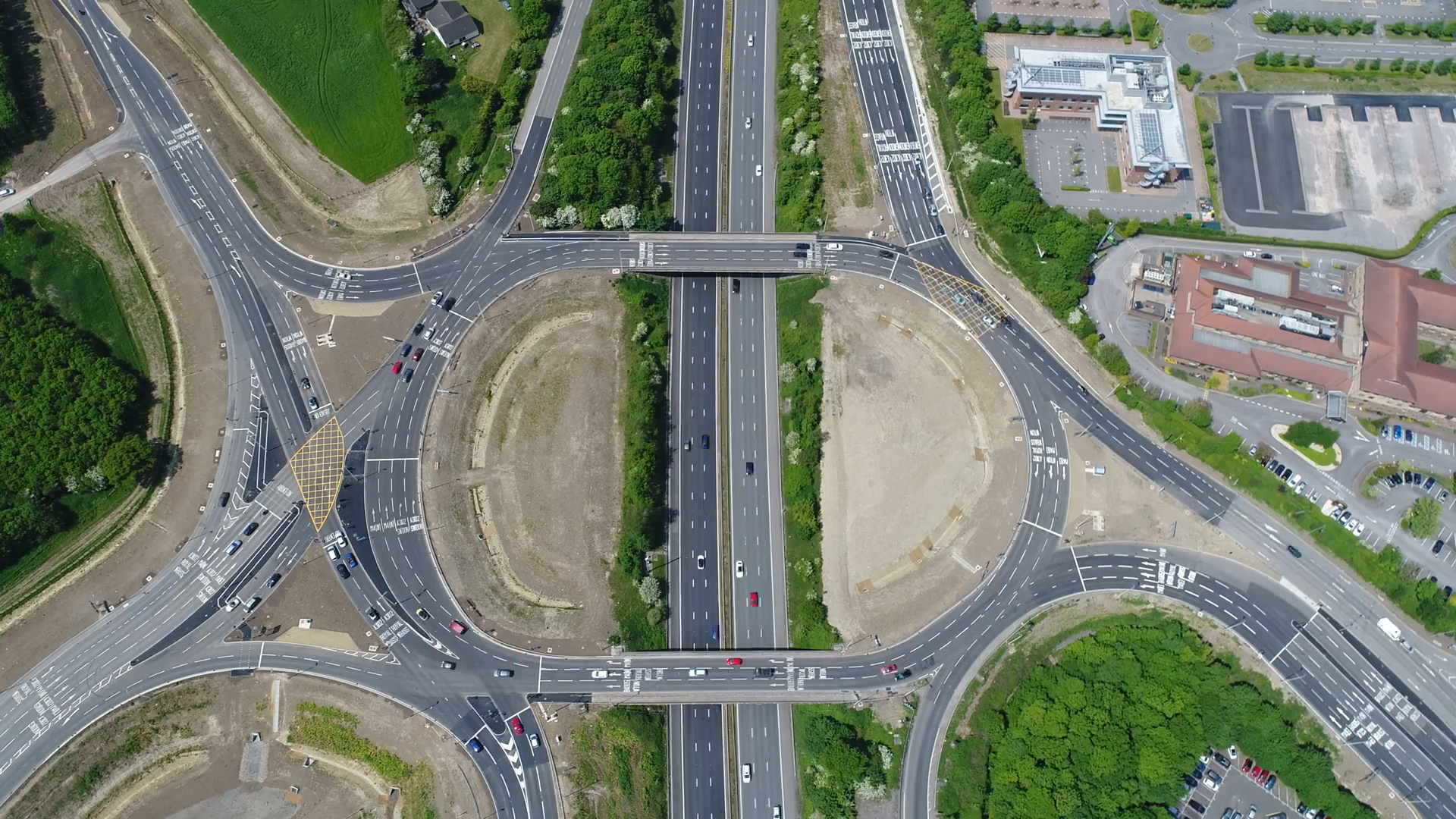}
    \caption{A sample image from the selected dataset}
    \label{fig:dataset_scenario}
\end{figure}

The performance of the proposed interaction-aware tracking method has been compared with the LMB filter, using the same parameters described above. Both methods use number of particles $N=200$, and tracks with probabilities of existence lower than $ 10^{-4}$ have been pruned. Owing to a large number of vehicles present in each frame and the size of each target vehicle being very small, interaction identification proved to be very challenging. Therefore, in the results presented, we have not implemented the step presented in Fig. ~\ref{fig:front_or_beihnd} for interaction identification, as it resulted in some adjacent lane vehicles identified as interacting with each other.  Performance evaluation has been conducted using a number of metrics, first of which is the optimal sub-pattern assignment (OSPA). The OSPA metric measures the error between ground truth and estimated tracks. It has three components, the overall OSPA error, the localisation error and cardinality error. The OSPA metric has two parameters, i.e. the cut-off for error $c$, which has been set to $100$ and the order $p$, which has been set to $2$ for the purpose of performance evaluation. It can be seen from Fig.~\ref{fig:OSPA_comparison} that the proposed method performs better than the LMB filter for most of the time steps. Since most of the error values are seen for the overall OSPA error and the OSPA cardinality, we have calculated an average error difference for these for the compared methods. The overall average of OSPA error difference is calculated to be $1.666$ and OSPA-cardinality error as $2.108$.
\begin{figure*}
    \centering
    \includegraphics[width=0.75\textwidth]{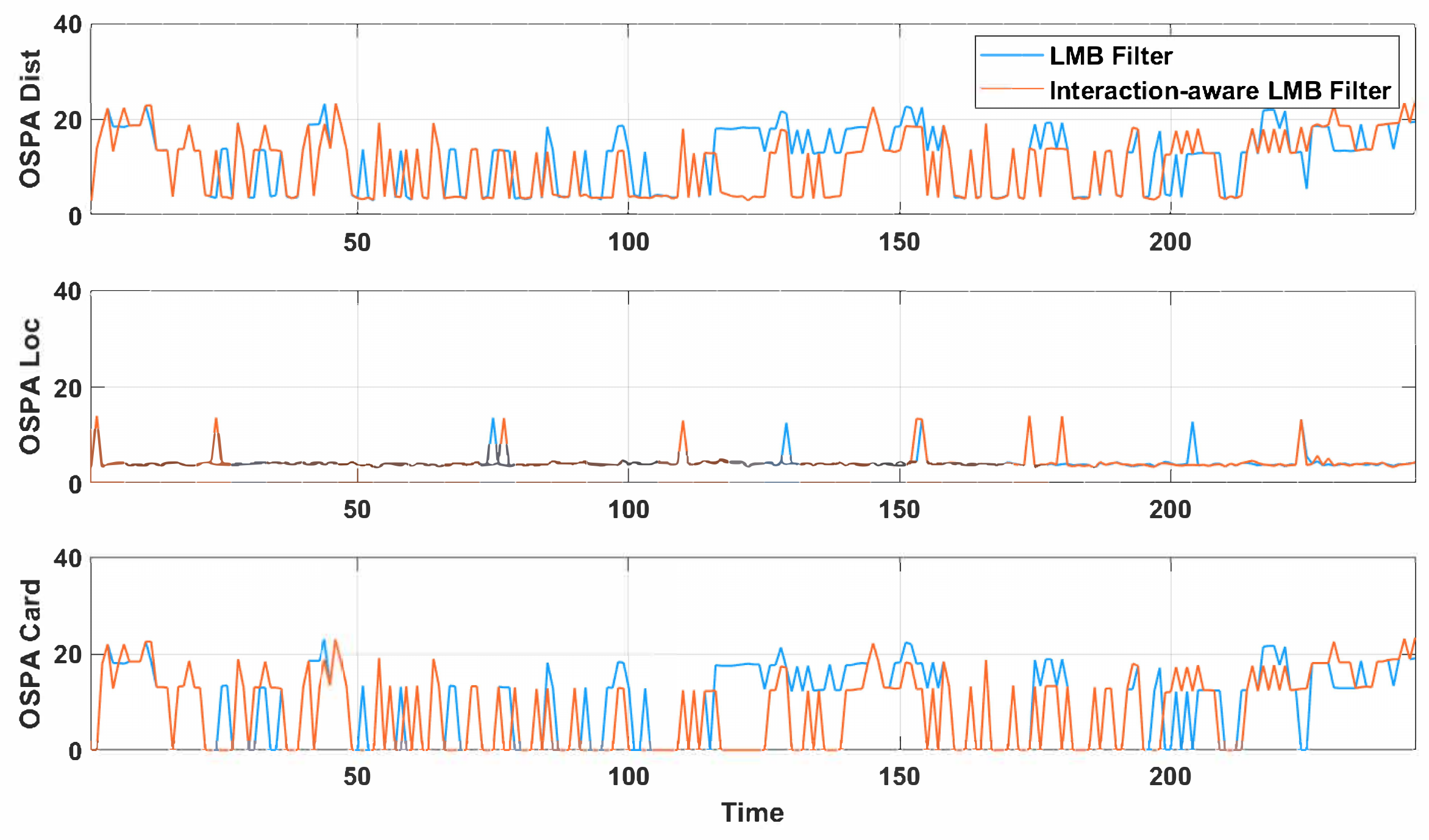}
    \caption{A comparison of OSPA metric for LMB filter and the proposed method}
    \label{fig:OSPA_comparison}
\end{figure*}
Another commonly used metric for evaluation of tracking performance is the $\text{OSPA}^{(2)}$, which calculates cumulative error over a specific window length. In addition to the cut-off and order parameters, which have been kept same as OSPA metric for consistency, the $\text{OSPA}^{(2)}$ window length has been set to $5$. It can be seen from the $\text{OSPA}^{(2)}$ plots in Fig.~\ref{fig:OSPA2_comparison} that the proposed method outperforms the LMB filter. The overall average error difference for LMB and proposed method for $\text{OSPA}^{(2)}$ is $0.195$ and for $\text{OSPA}^{(2)}$ cardinality is $-0.115$. 
\begin{figure*}
    \centering
    \includegraphics[width=0.75\textwidth]{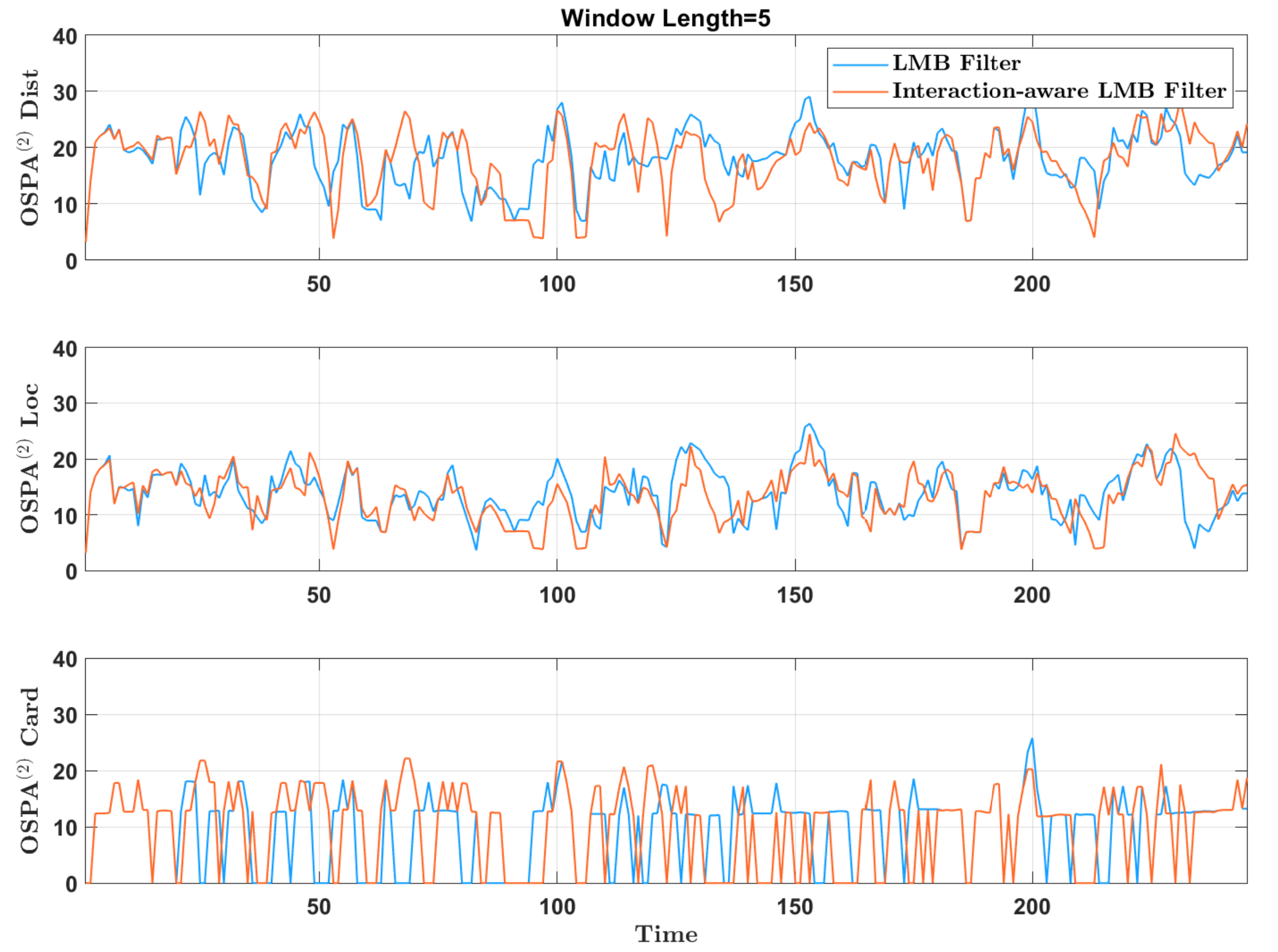}
    \caption{A comparison of OSPA-2 metric for LMB filter and the proposed method}
    \label{fig:OSPA2_comparison}
\end{figure*}
In order to further depict that interaction-aware method works well, we have plotted a comparison of the cardinality error for all time frames, where the cardinality error has been calculated as:
$$
\text{Card. Error} = \text{Ground Truth Card.} - \text{Estimated Card.}
$$

\begin{figure}
    \centering
    \includegraphics[angle=270,width=0.99\columnwidth]{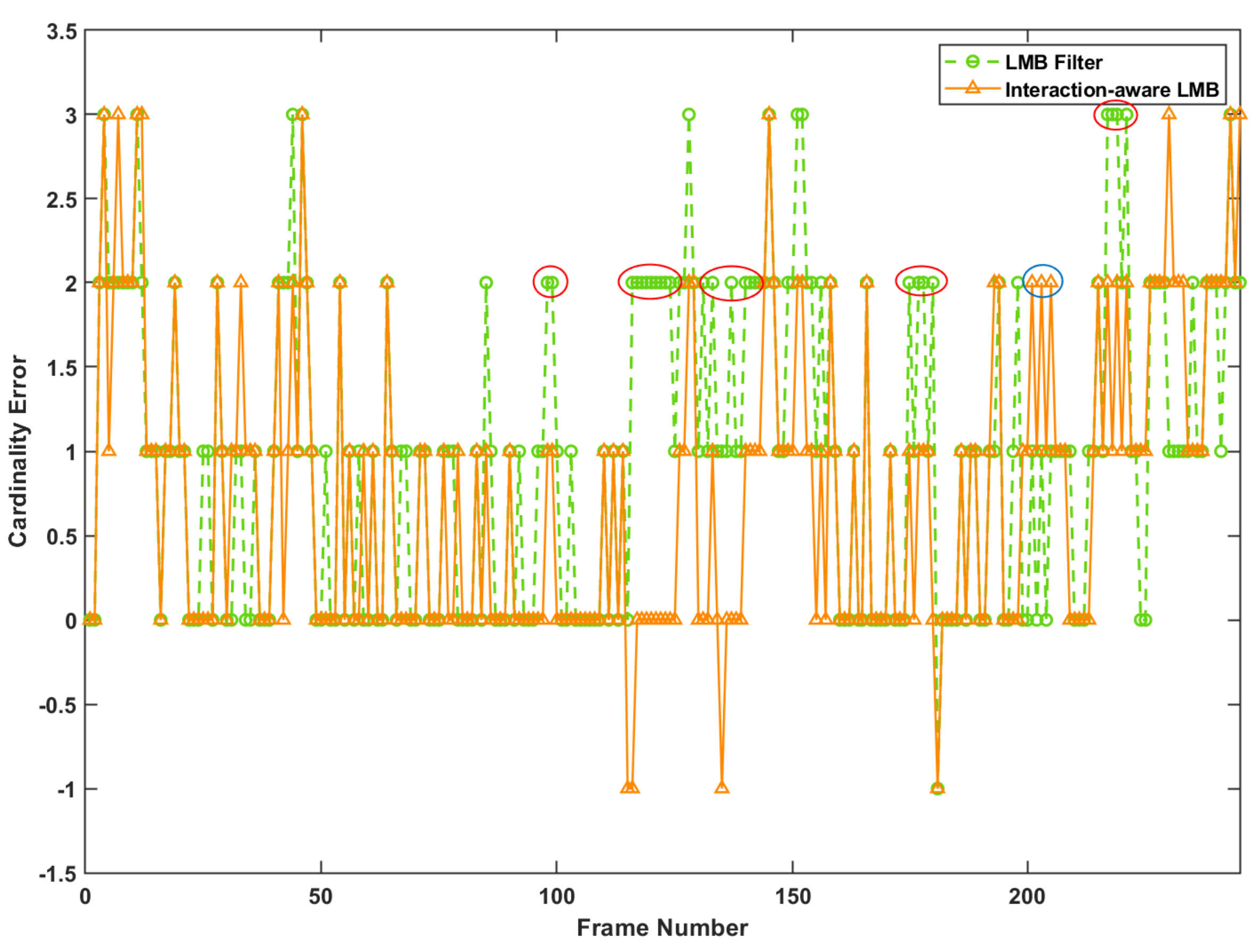}
    \caption{A comparison of cardinality error for LMB filter and the proposed method}
    \label{fig:Cardinality_error_comparison}
\end{figure}

The cardinality errors for LMB filter and the proposed method are depicted in Fig.~\ref{fig:Cardinality_error_comparison}. The portions of the plot highlighted with red ellipses show examples of the case when the cardinality error for LMB filter is higher than the proposed method, while the blue color ellipses show samples of cases where our method has a higher cardinality error. It can be clearly seen that our proposed interaction-aware LMB filter returns a more accurate cardinality than the traditional LMB filter in a large portion of the time frames.  The cardinality error highlights the fact that while the LMB filter sometimes loses a track (which may be picked up again by the filter at a later time), our method is far less prone to this issue. The interaction-aware LMB filter has over-estimation for a few frames, depicted by negative cardinality error. However, it should be noted that the LMB filter has a greater cardinality error for the set of measurements used than the interaction-aware LMB filter. Overall, it can be seen from the metrics shown above that the interaction-aware LMB filter has generally improved performance even for the selected scenario, which is quite complex due to the large and varying number of targets, the number of maneuvers and the varying speeds of targets in different areas of the scene. The performance difference is highlighted most in terms of cardinality error.

\section{Conclusion}
\label{sec:conc}
A novel approach was presented to incorporate target interactions into the prediction step of a RFS-based multi-target filter. The filter of choice in this paper is the LMB filter, but the proposed approach can be directly formulated into other similar filters such as the $\delta$-GLMB filter. The main idea was to incorporate interactions in the form of extra parameters involved in the single-target state transition density, with the parameters being time-varying, chosen as parts of the multi-target density information from the previous time step. 

In two practical examples, we elaborated further on how target interactions can be incorporated into the SMC implementation of the LMB filter, resulting in particle weight changes in the prediction step of the filter. We presented a challenging tracking scenario in which a large number of vehicles are moving and interacting in a complex multi-way intersection. The results demonstrated how the incorporation of interaction into the filter improves the tracking results in terms of both the OSPA and OSPA$^{(2)}$ error metrics. 

In general, this paper clearly demonstrates how the incorporation of information into the Bayesian filtering process can improve the estimation of states and labels of multiple targets in ITS applications. In essence, the RFS filters were originally invented with this purpose in mind. They provided a mathematically solid way to incorporate all target- and scene-related information into the prediction step and all measurement-related information into the update step of the filter. The motion model, the probability of survival, and the birth model are examples of the target- and scene-related information. The likelihood function, the probability of detection and the clutter model are examples of measurement-related information. In our work, we added the information that are available regarding interactions between targets (emphasizing that they are target-related information) into the prediction step of the most advanced RFS filter. This could be further enhanced by incorporating more information such as the road information (scene-related) into the prediction step as well. Indeed, we are currently working on this as a future area of further research.

\section*{Acknowledgment}

This work was supported by the Australian Research Council through Grant DE210101181.

\ifCLASSOPTIONcaptionsoff
  \newpage
\fi



%

\bibliographystyle{IEEEtran}
\bibliography{refference}

\begin{thebibliography}{10}
\providecommand{\url}[1]{#1}
\csname url@samestyle\endcsname
\providecommand{\newblock}{\relax}
\providecommand{\bibinfo}[2]{#2}
\providecommand{\BIBentrySTDinterwordspacing}{\spaceskip=0pt\relax}
\providecommand{\BIBentryALTinterwordstretchfactor}{4}
\providecommand{\BIBentryALTinterwordspacing}{\spaceskip=\fontdimen2\font plus
\BIBentryALTinterwordstretchfactor\fontdimen3\font minus
  \fontdimen4\font\relax}
\providecommand{\BIBforeignlanguage}[2]{{%
\expandafter\ifx\csname l@#1\endcsname\relax
\typeout{** WARNING: IEEEtran.bst: No hyphenation pattern has been}%
\typeout{** loaded for the language `#1'. Using the pattern for}%
\typeout{** the default language instead.}%
\else
\language=\csname l@#1\endcsname
\fi
#2}}
\providecommand{\BIBdecl}{\relax}
\BIBdecl

\bibitem{daronkolaei2008joint}
A.~G. Daronkolaei, V.~Nazari, M.~B. Menhaj, and S.~Shiry, ``A joint probability
  data association filter algorithm for multiple robot tracking problems,''
  \emph{Tools in Artificial Intelligence}, pp. 163--186, 2008.

\bibitem{mahler_book}
R.~P.~S. Mahler, \emph{Statistical multisource-multitarget information fusion},
  ser. Electronic Warfare.\hskip 1em plus 0.5em minus 0.4em\relax Norwood, MA,
  USA: Artech House, 2007.

\bibitem{mahler_phd2003}
------, ``Multitarget bayes filtering via first-order multitarget moments,''
  \emph{IEEE Transactions on Aerospace and Electronic Systems}, vol.~39, no.~4,
  pp. 1152--1178, 2003.

\bibitem{MahlerCPHD}
R.~Mahler, ``{PHD} filters of higher order in target number,'' \emph{IEEE
  Trans. Aerospace {\&} Electronic Systems}, vol.~43, no.~4, pp. 1523--1543,
  2007.

\bibitem{Vo_2005}
B.-N.~N. Vo, S.~Singh, and A.~Doucet, ``Sequential {M}onte {C}arlo methods for
  multitarget filtering with random finite sets,'' \emph{IEEE Transactions on
  Aerospace and Electronic Systems}, vol.~41, no.~4, pp. 1224--1245, oct 2005.

\bibitem{vo_cphd}
B.~N.~T. Vo, B.~N.~T. Vo, and A.~Cantoni, ``Analytic implementations of the
  cardinalized probability hypothesis density filter,'' \emph{IEEE Transactions
  on Signal Processing}, vol.~55, no. 7 II, pp. 3553--3567, 2007.

\bibitem{Vo_FW_Smoothing_TAES_2012}
R.~P.~S. Mahler, B.~N. Vo, and B.~T. Vo, ``Multi-target forward backward
  smoothing with the probability hypothesis density,'' \emph{IEEE Transactions
  on Aerospace and Electronic Systems}, vol.~48, no.~1, pp. 707--728, 2012.

\bibitem{vo_gmphd}
B.-N. Vo and W.-K. Ma, ``The {G}aussian mixture probability hypothesis density
  filter,'' \emph{IEEE Transactions on Signal Processing}, vol.~54, no.~11, pp.
  4091--4104, 2006.

\bibitem{Vo_SMC_PHD}
B.~N. Vo, S.~Singh, and A.~Doucet, ``Sequential monte carlo methods for
  multi-target filtering with random finite sets,'' \emph{IEEE Transactions on
  Aerospace and Electronic Systems}, vol.~41, no.~4, pp. 1224--1245, 2005.

\bibitem{MeMBer_Vo2}
B.~T. Vo, B.~N. Vo, and A.~Cantoni, ``The cardinality balanced multi-target
  multi-{B}ernoulli filter and its implementations,'' \emph{IEEE Transactions
  on Signal Processing}, vol.~57, no.~2, pp. 409--423, 2009.

\bibitem{Reza_2012}
\BIBentryALTinterwordspacing
R.~Hoseinnezhad, B.-N. Vo, B.-T. Vo, and D.~Suter, ``Visual tracking of
  numerous targets via multi-{B}ernoulli filtering of image data,''
  \emph{Pattern Recognition}, vol.~45, no.~10, pp. 3625 -- 3635, 2012.
  [Online]. Available:
  \url{http://www.sciencedirect.com/science/article/pii/S0031320312001616}
\BIBentrySTDinterwordspacing

\bibitem{Reza_audio_visual}
------, ``{B}ayesian integration of audio and visual information for
  multi-target tracking using a {CB-MeMBer} filter,'' in \emph{ICASSP'2001},
  2011, pp. 2300--2303.

\bibitem{Reza_visual_tracking}
R.~{Hoseinnezhad}, B.~{Vo}, and B.~{Vo}, ``Visual tracking in background
  subtracted image sequences via multi-{B}ernoulli filtering,'' \emph{IEEE
  Transactions on Signal Processing}, vol.~61, no.~2, pp. 392--397, Jan 2013.

\bibitem{Vo2013}
B.~T. Vo and B.~N. Vo, ``Labeled random finite sets and multi-object conjugate
  priors,'' \emph{IEEE Transactions on Signal Processing}, vol.~61, no.~13, pp.
  3460--3475, 2013.

\bibitem{LMB_Vo2}
B.~N. Vo, B.~T. Vo, and D.~Phung, ``Labeled random finite sets and the {B}ayes
  multi-target tracking filter,'' \emph{IEEE Transactions on Signal
  Processing}, vol.~62, no.~24, pp. 6554--6567, 2014.

\bibitem{reuter2014labeled}
S.~Reuter, B.-T. Vo, B.-N. Vo, and K.~Dietmayer, ``The labeled
  multi-{{B}ernoulli} filter,'' \emph{IEEE Transactions on Signal Processing},
  vol.~62, no.~12, pp. 3246--3260, 2014.

\bibitem{8531656}
S.~{Li}, G.~{Battistelli}, L.~{Chisci}, W.~{Yi}, B.~{Wang}, and L.~{Kong},
  ``Computationally efficient multi-agent multi-object tracking with labeled
  random finite sets,'' \emph{IEEE Transactions on Signal Processing}, vol.~67,
  no.~1, pp. 260--275, Jan 2019.

\bibitem{papi2015generalized}
F.~Papi, B.-N. Vo, B.-T. Vo, C.~Fantacci, and M.~Beard, ``Generalized labeled
  multi-bernoulli approximation of multi-object densities,'' \emph{IEEE
  Transactions on Signal Processing}, vol.~63, no.~20, pp. 5487--5497, 2015.

\bibitem{deepLearningModelBased}
J.~Pinto, G.~Hess, W.~Ljungbergh, Y.~Xia, H.~Wymeersch, and L.~Svensson, ``Can
  deep learning be applied to model-based multi-object tracking?'' \emph{arXiv
  preprint arXiv:2202.07909}, 2022.

\bibitem{deepLearningSurvey2020}
S.~Dargan, M.~Kumar, M.~R. Ayyagari, and G.~Kumar, ``A survey of deep learning
  and its applications: a new paradigm to machine learning,'' \emph{Archives of
  Computational Methods in Engineering}, vol.~27, no.~4, pp. 1071--1092, 2020.

\bibitem{conti2019soft}
A.~Conti, S.~Mazuelas, S.~Bartoletti, W.~C. Lindsey, and M.~Z. Win, ``Soft
  information for localization-of-things,'' \emph{Proceedings of the IEEE},
  vol. 107, no.~11, pp. 2240--2264, 2019.

\bibitem{conti2021location}
A.~Conti, F.~Morselli, Z.~Liu, S.~Bartoletti, S.~Mazuelas, W.~C. Lindsey, and
  M.~Z. Win, ``Location awareness in beyond 5g networks,'' \emph{IEEE
  Communications Magazine}, vol.~59, no.~11, pp. 22--27, 2021.

\bibitem{saucan2020information}
A.~A. Saucan and M.~Z. Win, ``Information-seeking sensor selection for
  ocean-of-things,'' \emph{IEEE Internet of Things Journal}, vol.~7, no.~10,
  pp. 10\,072--10\,088, 2020.

\bibitem{win2018theoretical}
M.~Z. Win, Y.~Shen, and W.~Dai, ``A theoretical foundation of network
  localization and navigation,'' \emph{Proceedings of the IEEE}, vol. 106,
  no.~7, pp. 1136--1165, 2018.

\bibitem{bartoletti2014mathematical}
S.~Bartoletti, W.~Dai, A.~Conti, and M.~Z. Win, ``A mathematical model for
  wideband ranging,'' \emph{IEEE Journal of Selected Topics in Signal
  Processing}, vol.~9, no.~2, pp. 216--228, 2014.

\bibitem{meyer2018message}
F.~Meyer, T.~Kropfreiter, J.~L. Williams, R.~Lau, F.~Hlawatsch, P.~Braca, and
  M.~Z. Win, ``Message passing algorithms for scalable multitarget tracking,''
  \emph{Proceedings of the IEEE}, vol. 106, no.~2, pp. 221--259, 2018.

\bibitem{group_tracking1}
S.~K. Pang, J.~Li, and S.~J. Godsill, ``Detection and tracking of coordinated
  groups,'' \emph{IEEE Transactions on Aerospace and Electronic Systems},
  vol.~47, no.~1, pp. 472--502, 2011.

\bibitem{group_tracking2}
A.~Gning, L.~Mihaylova, S.~Maskell, S.~K. Pang, and S.~Godsill, ``Group object
  structure and state estimation with evolving networks and monte carlo
  methods,'' \emph{IEEE Transactions on Signal Processing}, vol.~59, no.~4, pp.
  1383--1396, 2010.

\bibitem{gardner1996interactive}
W.~F. Gardner and D.~T. Lawton, ``Interactive model-based vehicle tracking,''
  \emph{IEEE Transactions on Pattern Analysis and Machine Intelligence},
  vol.~18, no.~11, pp. 1115--1121, 1996.

\bibitem{yang2018vehicle}
Z.~Yang and L.~S. Pun-Cheng, ``Vehicle detection in intelligent transportation
  systems and its applications under varying environments: A review,''
  \emph{Image and Vision Computing}, vol.~69, pp. 143--154, 2018.

\bibitem{kesting2007general}
A.~Kesting, M.~Treiber, and D.~Helbing, ``General lane-changing model mobil for
  car-following models,'' \emph{Transportation Research Record}, vol. 1999,
  no.~1, pp. 86--94, 2007.

\bibitem{rahman2013review}
M.~Rahman, M.~Chowdhury, Y.~Xie, and Y.~He, ``Review of microscopic
  lane-changing models and future research opportunities,'' \emph{IEEE
  transactions on intelligent transportation systems}, vol.~14, no.~4, pp.
  1942--1956, 2013.

\bibitem{saifuzzaman2014incorporating}
M.~Saifuzzaman and Z.~Zheng, ``Incorporating human-factors in car-following
  models: a review of recent developments and research needs,''
  \emph{Transportation research part C: emerging technologies}, vol.~48, pp.
  379--403, 2014.

\bibitem{song2017multi}
D.~Song, R.~Tharmarasa, T.~Kirubarajan, and X.~N. Fernando, ``Multi-vehicle
  tracking with road maps and car-following models,'' \emph{IEEE Transactions
  on Intelligent Transportation Systems}, vol.~19, no.~5, pp. 1375--1386, 2017.

\bibitem{pellegrini2009you}
S.~Pellegrini, A.~Ess, K.~Schindler, and L.~Van~Gool, ``You'll never walk
  alone: {M}odeling social behavior for multi-target tracking,'' in \emph{2009
  IEEE 12th International Conference on Computer Vision}.\hskip 1em plus 0.5em
  minus 0.4em\relax IEEE, 2009, pp. 261--268.

\bibitem{pellegrini2010improving}
S.~Pellegrini, A.~Ess, and L.~{Van Gool}, ``Improving data association by joint
  modeling of pedestrian trajectories and groupings,'' in \emph{Lecture Notes
  in Computer Science (including subseries Lecture Notes in Artificial
  Intelligence and Lecture Notes in Bioinformatics)}, vol. 6311 LNCS, no. PART
  1, 2010, pp. 452--465.

\bibitem{yuan2017tracking}
Y.~Yuan, Y.~Lu, and Q.~Wang, ``Tracking as a whole: Multi-target tracking by
  modeling group behavior with sequential detection,'' \emph{IEEE Transactions
  on Intelligent Transportation Systems}, vol.~18, no.~12, pp. 3339--3349,
  2017.

\bibitem{khalkhali2019multi}
M.~B. Khalkhali, A.~Vahedian, and H.~S. Yazdi, ``Multi-target state estimation
  using interactive kalman filter for multi-vehicle tracking,'' \emph{IEEE
  Transactions on Intelligent Transportation Systems}, vol.~21, no.~3, pp.
  1131--1144, 2019.

\bibitem{leven2009unscented}
W.~F. Leven and A.~D. Lanterman, ``Unscented kalman filters for multiple target
  tracking with symmetric measurement equations,'' \emph{IEEE Transactions on
  Automatic Control}, vol.~54, no.~2, pp. 370--375, 2009.

\bibitem{MOT_app_surv1}
A.~Roy and D.~Mitra, ``Unscented-kalman-filter-based multitarget tracking
  algorithms for airborne surveillance application,'' \emph{Journal of
  Guidance, Control, and Dynamics}, vol.~39, no.~9, pp. 1949--1966, 2016.

\bibitem{chen2001tracking}
B.~Chen and J.~K. Tugnait, ``Tracking of multiple maneuvering targets in
  clutter using imm/jpda filtering and fixed-lag smoothing,''
  \emph{Automatica}, vol.~37, no.~2, pp. 239--249, 2001.

\bibitem{Multiple_Model_GLMB}
W.~{Yi}, M.~{Jiang}, and R.~{Hoseinnezhad}, ``The multiple model vo–vo
  filter,'' \emph{IEEE Transactions on Aerospace and Electronic Systems},
  vol.~53, no.~2, pp. 1045--1054, 2017.

\bibitem{our_icassp_2010}
R.~Hoseinnezhad, B.~N. Vo, D.~Suter, and B.~T. Vo, ``Multi-object filtering
  from image sequence without detection,'' in \emph{ICASSP, IEEE International
  Conference on Acoustics, Speech and Signal Processing - Proceedings}, Dallas,
  TX, Mar 2010, pp. 1154--1157.

\bibitem{MOT_application_Vision002}
T.~Rathnayake, R.~Tennakoon, A.~K.~Gostar, A.~Bab-Hadiashar, and
  R.~Hoseinnezhad, ``Information fusion for industrial mobile platform safety
  via track-before-detect labeled multi-{B}ernoulli filter,'' \emph{Sensors},
  vol.~19, no.~9, p. 2016, 2019.

\bibitem{Battistelli2013}
G.~Battistelli, L.~Chisci, C.~Fantacci, A.~Farina, and A.~Graziano, ``Consensus
  {CPHD} filter for distributed multitarget tracking,'' \emph{IEEE Journal on
  Selected Topics in Signal Processing}, vol.~7, no.~3, pp. 508--520, 2013.

\bibitem{GCI-MB}
B.~L. Wang, W.~Yi, R.~Hoseinnezhad, S.~Q. Li, L.~J. Kong, and X.~B. Yang,
  ``Distributed fusion with multi-{B}ernoulli filter based on generalized
  {C}ovariance {I}ntersection,'' \emph{IEEE Trans. Signal Process.}, vol.~65,
  no.~1, pp. 242--255, 2017.

\bibitem{Fantacci-BT}
C.~Fantacci, B.~Vo, B.~Vo, G.~Battistelli, and L.~Chisci, ``Robust fusion for
  multisensor multiobject tracking,'' \emph{IEEE Signal Processing Letters},
  vol.~25, no.~5, pp. 640--644, may 2018.

\bibitem{Gostar_CSD_Fusion_SP2019}
A.~K. Gostar, T.~Rathnayake, R.~Tennakoon, A.~Bab-Hadiashar, G.~Battistelli,
  L.~Chisci, and R.~Hoseinnezhad, ``Cooperative sensor fusion in centralized
  sensor networks using {C}auchy–{S}chwarz divergence,'' \emph{Signal
  Processing}, vol. 167, p. 107278, 2020.

\bibitem{Li_eta_al_TSP_2018_robust_fusion}
S.~Li, W.~Yi, R.~Hoseinnezhad, G.~Battistelli, B.~Wang, and L.~Kong, ``Robust
  distributed fusion with labeled random finite sets,'' \emph{IEEE Transactions
  on Signal Processing}, vol.~66, no.~2, pp. 278--293, 2018.

\bibitem{li2018partial}
T.~{Li}, J.~M. {Corchado}, and S.~{Sun}, ``Partial consensus and conservative
  fusion of gaussian mixtures for distributed phd fusion,'' \emph{IEEE
  Transactions on Aerospace and Electronic Systems}, vol.~55, no.~5, pp.
  2150--2163, 2019.

\bibitem{li2019cardinality}
T.~{Li}, F.~{Hlawatsch}, and P.~M. {Djurić}, ``Cardinality-consensus-based phd
  filtering for distributed multitarget tracking,'' \emph{IEEE Signal
  Processing Letters}, vol.~26, no.~1, pp. 49--53, 2019.

\bibitem{LI2020107246}
\BIBentryALTinterwordspacing
G.~Li, G.~Battistelli, W.~Yi, and L.~Kong, ``Distributed multi-sensor
  multi-view fusion based on generalized covariance intersection,''
  \emph{Signal Processing}, vol. 166, p. 107246, 2020. [Online]. Available:
  \url{http://www.sciencedirect.com/science/article/pii/S0165168419302920}
\BIBentrySTDinterwordspacing

\bibitem{Wang2018}
X.~Wang, A.~K. Gostar, T.~Rathnayake, B.~Xu, A.~Bab-Hadiashar, and
  R.~Hoseinnezhad, ``Centralized multiple-view sensor fusion using labeled
  multi-{B}ernoulli filters,'' \emph{Signal Processing}, vol. 150, pp. 75--84,
  2018.

\bibitem{Gostar_CSD_LMB}
A.~K. Gostar, R.~Hoseinnezhad, T.~Rathnayake, X.~Wang, and A.~Bab-Hadiashar,
  ``Constrained sensor control for labeled multi-{B}ernoulli filter using
  {C}auchy-{S}chwarz divergence,'' \emph{IEEE Signal Processing Letters},
  vol.~24, no.~9, pp. 1313--1317, sep 2017.

\bibitem{Gostar2015_TAES}
A.~K. Gostar, R.~Hoseinnezhad, and A.~Bab-Hadiashar, ``Multi-{B}ernoulli
  sensor-control via minimization of expected estimation errors,'' \emph{IEEE
  Trans. AES}, vol.~51, no.~3, pp. 1762--1773, jul 2015.

\bibitem{gostar2016multi}
------, ``Multi-{B}ernoulli sensor-selection for multi-target tracking with
  unknown clutter and detection profiles,'' \emph{Signal Processing}, vol. 119,
  pp. 28--42, 2016.

\bibitem{gostar2017a}
A.~K. Gostar, R.~Hoseinnezhad, A.~Bab-Hadiashar, and W.~Liu,
  ``Sensor-management for multitarget filters via minimization of posterior
  dispersion,'' \emph{IEEE Transactions on Aerospace and Electronic Systems},
  vol.~53, no.~6, pp. 2877--2884, 2017.

\bibitem{gostar2017b}
A.~K. Gostar, R.~Hoseinnezhad, T.~Rathnayake, X.~Wang, and A.~Bab-Hadiashar,
  ``Constrained sensor control for labeled multi-{B}ernoulli filter using
  {C}auchy–{S}chwarz divergence,'' \emph{IEEE Signal Processing Letters},
  vol.~24, no.~9, pp. 1313--1317, 2017.

\bibitem{wang2018a}
X.~Wang, R.~Hoseinnezhad, A.~K. Gostar, T.~Rathnayake, B.~Xu, and
  A.~Bab-Hadiashar, ``Multi-sensor control for multi-object bayes filters,''
  \emph{Signal Processing}, vol. 142, pp. 260--270, 2018.

\bibitem{vo_OSPA_Trans_SP}
D.~Schuhmacher, B.~N.~T. Vo, and B.~N.~T. Vo, ``A consistent metric for
  performance evaluation of multi-object filters,'' \emph{IEEE Transactions on
  Signal Processing}, vol.~56, no. 8 I, pp. 3447--3457, 2008.

\bibitem{gostar2019interactive}
A.~K. Gostar, T.~Rathnayake, C.~Fu, A.~Bab-Hadiashar, G.~Battistelli,
  L.~Chisci, and R.~Hoseinnezhad, ``Interactive multiple-target tracking via
  labeled multi-bernoulli filters,'' in \emph{2019 International Conference on
  Control, Automation and Information Sciences (ICCAIS)}.\hskip 1em plus 0.5em
  minus 0.4em\relax IEEE, 2019, pp. 1--6.

\bibitem{li2003survey}
X.~R. Li and V.~P. Jilkov, ``Survey of maneuvering target tracking. part i.
  dynamic models,'' \emph{IEEE Transactions on aerospace and electronic
  systems}, vol.~39, no.~4, pp. 1333--1364, 2003.

\end{thebibliography}

%




\newpage
 \begin{IEEEbiography}[{\includegraphics[width=1in,height=1.25in,clip,keepaspectratio]{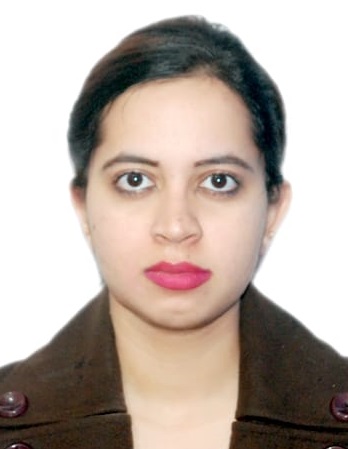}}]{Nida Ishtiaq}  completed her bachelor’s and masters’ degrees in Electrical Engineering from National University of Sciences and Technology (NUST) CEME, Rawalpindi, Pakistan, in 2013 and 2016, respectively. She has worked in research and academia in reputable universities in Pakistan for about 3 years. She is currently undertaking Ph.D. studies at RMIT University, Melbourne, Australia since 2019. Her research interests include signal processing and random-finite-set based multi-object tracking with special focus on vehicle tracking.
 \end{IEEEbiography}
 \begin{IEEEbiography}[{\includegraphics[width=1in,height=1.25in,clip,keepaspectratio]{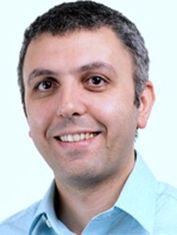}}]{Amirali Khodadadian Gostar} received his BSc degree in Electrical Engineering and MSc degree in Philosophy of Science, PhD degree in Mechatronics Engineering from RMIT University. He is currently a lecturer and ARC DECRA fellow in the School of Engineering at RMIT University. His research interests include sensor management, data fusion, and multitarget tracking.
 \end{IEEEbiography}
\begin{IEEEbiography}[{\includegraphics[width=1in,height=1.25in,clip,keepaspectratio]{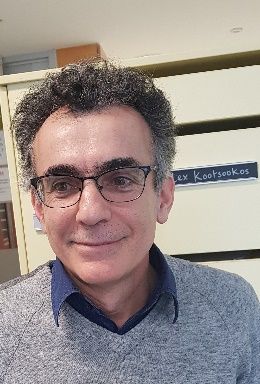}}]{Alireza Bab-Hadiashar} received his BSc and MEng in Mechanical
 Engineering, then PhD in Robotics from Monash University. He has
 held various positions in Monash University, Swinburne University of
 Technology and RMIT University where he is currently a professor of
 mechatronics and leads the intelligent automation research group.
 His main area of research interest is intelligent automation in general,
 and robust data fitting in machine vision, deep learning for detection
 and identification, and robust data segmentation, in particular
\end{IEEEbiography}
\begin{IEEEbiography}[{\includegraphics[width=1in,height=1.25in,clip,keepaspectratio]{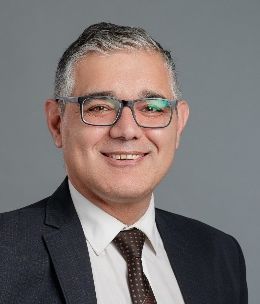}}]{Reza Hoseinnezhad} received his BSc, MSc, and PhD in Electrical
Engineering from University of Tehran, Iran, in 1994, 1996 and
2002, respectively. He has held various positions at University of
Tehran, Swinburne University of Technology, The University of
 Melbourne, and RMIT University where he has worked since 2010
 and is currently a Professor and Research Development Lead as
 well as the Discipline Leader (Manufacturing and Mechatronics) at
 School of Engineering. His main areas of research interest are
 statistical information fusion, random finite sets, multi-object
 tracking, deep learning, and robust multi-structure data fitting in
 computer vision.
 \end{IEEEbiography}
\end{document}